\documentclass[11pt]{llncs}
\usepackage[letterpaper,hmargin=1in,vmargin=1.25in]{geometry}

\usepackage{makeidx}  % allows for indexgeneration
\usepackage{color}
\usepackage{graphicx}
\usepackage{hyperref}
\usepackage{physics}
\usepackage{mathtools}

\usepackage{float}
\usepackage{amsmath}
\usepackage{enumitem}
\usepackage{algorithm}
\floatname{algorithm}{Protocol}

\DeclarePairedDelimiter\ceil{\lceil}{\rceil}

\def\EQ#1{\begin{eqnarray}#1\end{eqnarray}}
\def\H {{\mathcal H}}
\newcommand{\AR}[2][c]{$$\begin{array}[#1]{lllllllllllllll}#2\end{array}$$}
\begin{document}
\frontmatter          % for the preliminaries
\pagestyle{plain}  

\mainmatter              % start of the contributions

\title{The Quantum Cut-and-Choose Technique and Quantum Two-Party Computation}
\author{Elham Kashefi$^{1,2}$, Luka Music$^2$ and Petros Wallden$^1$}

\institute{1. School of Informatics, University of Edinburgh,10 Crichton Street, Edinburgh EH8 9AB, UK\\
2. Departement Informatique et Reseaux, UPMC - LIP6, 4 Place Jussieu 75252 Paris CEDEX 05, France}

\maketitle              

\begin{abstract}
The application and analysis of the Cut-and-Choose technique in protocols secure against quantum adversaries is not a straightforward transposition of the classical case, among other reasons due to the difficulty to use ``rewinding'' in the quantum realm. We introduce a Quantum Computation Cut-and-Choose (QC-CC) technique which is a generalisation of the classical Cut-and-Choose in order to build quantum protocols secure against \emph{quantum covert adversaries}. Such adversaries can essentially deviate arbitrarily provided that their deviation is not detected with high probability. As an application of the QC-CC technique we give a protocol for securely performing a two-party quantum computation with classical input and output. As a basis we use the concept of secure delegated quantum computing \cite{bfk}, and in particular the protocol for quantum garbled circuit computation of \cite{KW16} that has been proven secure against only a weak specious adversaries (defined in \cite{DNS10}). A unique property of these protocols is the separation between classical and quantum communications and the asymmetry between client and server, which enables us to sidestep the issues linked to quantum rewinding. This opens the possibility of applying the QC-CC %this 
technique to other quantum protocols that have this separation. In our proof of security we adapt %modify 
and use (at different parts of the proof) two quantum rewinding techniques, namely Watrous' oblivious quantum rewinding \cite{Wat09} and  Unruh's special quantum rewinding \cite{Unr12}. Our protocol achieves the same functionality as in the previous work on secure two-party quantum computing such as the one in \cite{DNS12}, however using the Cut-and-Choose technique on the protocol from \cite{KW16} leads to the following key improvements: (i) only one-way offline quantum communication is necessary , (ii) only one party (server) needs to have involved quantum technological abilities, (iii) only minimal extra cryptographic primitives are required, namely one oblivious transfer for each input bit and quantum-safe commitments.
\end{abstract}

\section{Introduction}
%%text

A key task in modern cryptography is to compute a function of many inputs given by different parties that do not trust each other and wish to maintain the privacy of their input. This is called secure multi-party computation (to name some examples: millionaire's problem, coin tossing, voting schemes, etc). 
The field started with the seminal paper of Yao \cite{Yao86}, where two-parties that do not trust each other (they are ``honest-but-curious'') compute a function of their joint inputs. This protocol was later made secure against malicious adversaries by employing standard (classical) techniques for boosting the security of honest-but-curious protocols to the malicious adversarial setting (e.g. using the GMW compiler as in \cite{GMW87}). Another such technique is the Cut-and-Choose technique first used in this context in \cite{LinPin07}.

The quantum analogue (secure two-party quantum computation, or 2PQC) involves the computation of a function using a quantum computer and was first examined in \cite{DNS10} for a quantum honest-but-curious adversaries (called specious) and later made secure against more malicious adversaries in \cite{DNS12}. The latter protocol, did not use any of the standard boosting classical techniques, but instead used a stepwise quantum authentication protocol, where two-ways online quantum communication was required. Moreover, both protocols use extra classical cryptographic primitives, which in the case of the malicious \cite{DNS12} is a full actively secure \emph{classical} two-party computation primitive. 

The use of classical boosting techniques (such as Cut-and-Choose) for quantum protocols is complicated not only because specific care is needed when defining quantum analogues, for example, for garbled circuits but also for technical reasons since the rewinding method for proving security cannot be directly used in quantum protocols (as demonstrated for zero knowledge proofs \cite{Wat09} and zero knowledge proofs of knowledge \cite{Unr12}).
%=========================

\paragraph{Our Contribution}
%%text

\begin{enumerate}

\item We introduce a Quantum Computation %apply the 
Cut-and-Choose technique. Application of this technique is made possible because of the unique decomposition of quantum computation into a classical control and a quantum resource in the measurement-based quantum computing models such as gate teleportation. This separation furthermore provides a platform for a client-server setting for secure delegated computing \cite{bfk}.

\item We give a protocol for 2PQC with classical input and output\footnote{Note, that even though we evaluate a classical function we still need quantum computation if this function cannot be efficiently computed (in the honest case) with a classical computer (e.g. functions that involve factoring) and thus classical techniques are not applicable.} which is secure against ``quantum covert'' adversaries, a notion of strong adversaries similar with the classical covert adversaries \cite{AumLin07} (see below for formal definition and motivation). We use the aforementioned QC-CC technique and address the subtleties in the security proof due to rewinding. Our protocol, which builds on the work of \cite{KW16} that gave a protocol for 2PQC secure against weak specious adversaries, resembles the original protocol by Yao \cite{Yao86} (e.g. asymmetry between the two parties) and in particular the one in \cite{LinPin07}, where Yao's protocol is boosted to the malicious case using the classical Cut-and-Choose technique. 

\item A key obstruction when using classical techniques for boosting the security of quantum protocols is that in general rewinding the \emph{quantum} adversary during the simulation is \emph{not} possible. There are two known cases where rewinding can be used for quantum adversaries, namely Watrous' oblivious rewinding \cite{Wat09} and Unruh's special rewinding \cite{Unr12}. We adapt and use both methods in different places in order to construct the simulators and prove the security of our protocol. This is one of the few protocols in which quantum rewinding is explicitly used and the only one, to our knowledge, that uses two %different 
types of quantum rewinding.

\item \cite{DNS12,DNS10} describe 2PQC protocols in a setting where the two parties are symmetric. Our protocol crucially differs in a number of points (other than using the Cut-and-Choose technique): (i) There is only one-way offline quantum communication between the parties, (ii) only one party (``server'') needs involved quantum technological abilities, while the other (``client'') only needs to prepare offline single qubits, (iii) minimal classical cryptographic primitives are required, namely oblivious transfer for input bits and quantum-safe commitments.
\end{enumerate}

\noindent %The paper is organised in the following way. 
In Section \ref{prelims} we present background material and introduce the notion of \emph{quantum covert adversaries} (similarly with \cite{AumLin07}). 
In Section \ref{quantum cc} we introduce the quantum-computation Cut-and-Choose technique. In Section \ref{prot} we give the protocol for 2PQC and %finally 
prove its security in Appendix \ref{sec proofs}.

\paragraph{Related works.} The field of secure two (and multi) party (classical) computation started with Yao's paper \cite{Yao86}, which was proven secure against malicious adversaries in \cite{GMW87} using generic Zero-Knowledge proofs and in \cite{LinPin07} with the Cut-and-Choose technique.  Covert adversaries were introduced in \cite{AumLin07} where again the Cut-and-Choose technique was used to achieve an even more efficient protocol. Yao's protocol has been used for a number of other functionalities, such as constructing non-interactive verifiable computing \cite{GGB10}. 
%%text

In the early days of quantum computation, researchers believed that quantum properties could lead to a breakthrough and achieve, with unconditional security, several (clasical) multi-party cryptographic primitives. However a series of no-go theorems, first proving that bit commitment is impossible \cite{commit1,commit2} , then oblivious transfer \cite{OT_nogo} and finally \cite{SSS09} showed that any non-trivial functionality leaks some information to adversaries. Since then, it is established that any such protocol is either only computationally secure or requires the existence of certain (quantum secure) simple cryptographic primitives.

Closely related is the question of what assumptions are required if one wants to perform a secure \emph{quantum} computation involving multiple parties. The case of 2PQC was addressed in \cite{DNS10} for quantum honest-but-curious and in \cite{DNS12} for malicious adversaries. The case of multiple parties was addressed in \cite{Ben-Or2006,multipartyqc} where an honest majority was required. 

In this work we use as basis the universal blind quantum computation protocol \cite{bfk} and its verifiable version \cite{fk}. For the case of weak specious adversaries, the 2PQC was addressed in \cite{KW16} while the multiparty quantum computation was also addressed in \cite{KP16}, again in a restricted setting. While we use measurement-based quantum computation (MBQC) \cite{onewaycomputer}, similar blind and verification protocols exist in the teleportation model \cite{B2015} and 2PQC or MPQC protocols could be explored for that case as well as other blind verification protocols such as \cite{hm2015}.

\section{Preliminaries and Security Definitions}\label{prelims}

\subsection{Verifiable Blind Quantum Computation}

The model for quantum computations used in our contribution is MBQC \cite{onewaycomputer}. In this section we introduce MBQC and revise protocols for blind quantum computation (server performs computation without learning input/output or computation) \cite{bfk} and verifiable blind quantum computation (client can also verify that the computation was performed correctly) \cite{fk} which are based on it.

The MBQC model of computation is equivalent to the circuit model as it is based %a derivation of 
on the gate teleportation principle. One starts with a large, generic entangled state (represented by a graph) and, by choosing suitable single qubit measurements, can perform any quantum computation (circuit). The computation is fully characterised by the graph and default measurement angles (see below) and is called a measurement pattern. %For completeness, 
See %we give 
an example %of mapping 
of a universal set of gates expressed as MBQC measurement patterns %can be found  
in Appendix \ref{example mbqc}. 

For our purpose it will be simpler to %we will 
consider a client-server setting. The client can prepare single qubits %to create single qubits in certain states, 
while the server can perform any general quantum computation. The client prepares and sends qubits in the $\ket{+}$ state and the server entangles them according to a certain computation graph by performing $\mathrm{controlled-}Z$ gates between all qubits corresponding to adjacent vertices on the graph, resulting in a \emph{graph state} (details can be found in \cite{hein2004multiparty}). The computation is defined by a default measurement angle $\phi_i$ (which depends only on the desired computation). It is carried out by having the server measure single qubits in an order defined by the flow. The actual angle of each measurement depends on $\phi_i$ and on the outcomes of previous measurements. In our setting, the client is responsible for these classical calculations to adjust the angle and therefore the server returns to the client the result of each measurement, for more details see \cite{bfk}.

Let $I$ and $O$ be respectively the sets of input and output qubits. A flow is defined by a function ($f : O^c \rightarrow I^c$) from measured qubits to non-input qubits and a partial order $(\preceq)$ over the vertices of the graph such that $\forall i, i \preceq f(i)$ and $\forall j \in N_G(i), f(i) \preceq j$, where $N_G(i)$ denotes the neighbours of $i$ in graph $G$. Each qubit $i$ is $X$-dependent on $X_i = f^{-1}(i)$ and $Z$-dependent on all qubits $j$ such that $i \in N_G(j)$ (this set is called $Z_i$). The existence of such a flow in all the graphs used for computations in MBQC patterns guarantees that the number of dependencies does not blow up. 
Given the sets $X_i$ and $Z_i$ the computation angle for qubit $i$ needs to be adjusted as such : let $s_i^X = \oplus_{j \in D_i^X} s_j$ and $s_i^Z = \oplus_{j \in D_i^Z} s_j$, where $s_j$ corresponds to the outcome of the measurement on qubit $j$ and $D_i^X$ and $D_i^Z$ are subsets of $X_i$ and $Z_i$ respectively (these qubits in $D_i^X$ and $D_i^Z$ have all already been measured as they belong to the past neighbours and past neighbours of past neighbours, see the flow construction in \cite{DK2006}). Then the corrected angle (the one that is actually measured) is $\phi'_i = (-1)^{s_i^X}\phi_i + s_i^Z\pi$.

The computation can be totally hidden from the server due to the following observation: if instead of sending $\ket{+}$ states the client chooses at random and sends $\ket{+_\theta}=1/\sqrt{2}(\ket{0}+e^{i\theta}\ket{1})$ with $\theta \in \{0, \pi/8, 2\cdot\pi/8,\ldots, 7\cdot\pi/8\}$ %without revealing the angle $\theta$, 
then measuring the qubits in a similarly rotated basis has the same result as the initial non-rotated computation. If the client keeps the angle $\theta$ hidden from the server, the server is completely blind on what computation is being performed. To ensure that no information is leaked from the measurement outcome,  
we add another parameter $r_i$ for each qubit, which serves as a One-Time-Pad for the measurement outcome. The resulting measurement angle with all parameters taken into account is then $\delta_i = C(\phi_i, s_i^X, s_i^Z, \theta_i, r_i) = \phi'_i(\phi_i, s_i^X, s_i^Z) + \theta_i + r_i\pi$. In short, the client sends rotated qubits (to become the resource state once entangled) and then guides the computation with a set of classical instructions. It is the combination of these two parts (quantum state preparation and classical instructions) that leads to the desired blind computation. This idea was formalised in the \emph{universal blind quantum computation} (UBQC) protocol in \cite{bfk}. 

We denote $ran_i$ collectively all the parameters that $\delta_i$ depends on other than $\phi_i$. We use bold and suppress the subscript that refers to a particular qubit in the graph to denote a full string, e.g. $\boldsymbol{\phi}:=\{\phi_1,\cdots,\phi_N\},\mathbf{ran}=\{ran_1,\cdots,ran_N\}$.
We then define the past of qubits, which allows us to calculate upper bounds on the number of dependencies for various computation graphs.

\begin{definition}[Past of qubit $i$ and Influence-past of qubit $i$] We define $P_i = Z_i \cup X_i$ to be the set of qubits $j$ that have $X$ or $Z$ dependency on $i$. We define influence-past $c_i$ of qubit $i$ to be an assignment of an outcome $b_j \in \{0, 1\}$ for all qubits $j \in P_i$. 
\end{definition}

For example the brickwork state \cite{bfk}, which can be used for universal quantum computation, has for each qubit a single X-dependent qubit and at most two Z-dependent qubits, and so the cardinality $\abs{P_i}$ is at most $3$ for all $i$. To each influence-past $c_i$ corresponds a unique value of $\delta_i$ (the corrected measurement angle for this influence-past). Note that we will denote $\boldsymbol{\delta}$ the set of instructions that include the measurement angles of each qubit, for all alternative influence-pasts, i.e. this is not a string of $N$ measurement angles $\delta_i$ but of $\sum_{i=1}^N|P_i|$ measurement angles of the form $\delta_i(c_i)$, where we have $|P_i|$ angles per qubit.

In UBQC, the server is not forced to follow the instructions and the client cannot verify if the computation is done correctly. One can modify the protocol to allow for such verification (see Theorem \ref{vbqc verif}),
as was first done in \cite{fk}. The central idea is to include trap qubits at positions unknown to the server. The client can send states from $\{\ket{0},\ket{1}\}$ (called dummies), which have the effect of breaking the graph at this vertex, removing it along with any attached edges. This can be used to generate isolated qubits in the graph in a way that is undetectable by the server. These isolated qubits do not affect the computation while they have deterministic outcome if measured in the correct basis. They can therefore be used as traps: a client can easily detect if one of them has been measured incorrectly but the server is ignorant of their position in the graph. This idea was introduced in \cite{fk} and later optimised by different protocols, such as \cite{KW15} which we use here. The reason to use \cite{KW15}, other than efficiency, is because the construction is ``local'' and the server can obtain some information about the true graph (needed for 2PQC) without compromising the security. %The overhead is linear in the size of the computation and the number of trap qubits is the same as the number of computation qubits. 
The construction of the resource %(see \cite{KW15}) 
given a base-graph $G$ (graph that the UBQC computation without traps requires), that has vertices $v\in V(G)$ and edges $e\in E(G)$, is the following: %summarised below:

\begin{enumerate}
\item For each vertex $v_i$, we define a set of three new vertices $P_{v_i}=\{p^{v_i}_1,p^{v_i}_3,p^{v_i}_3\}$. These are called \emph{primary} vertices.

\item Corresponding to each edge $e(v_i,v_j)\in E(G)$ of the base-graph that connects the base vertices $v_i$ and $v_j$, we introduce a set of nine edges $E_{e(v_i,v_j)}$ that connect each of the vertices in the set $P_{v_i}$ with each of the vertices in the set $P_{v_j}$.

\item We replace every edge in the resulting graph with a new vertex connected to the two vertices originally joined by that edge. The new vertices added in this step are called \emph{added} vertices. This is the \emph{dotted triple-graph} $DT(G)$.
\end{enumerate}

The edge or vertex of the initial graph that each vertex $v\in DT(G)$ belongs to is called its \emph{base-location}. We can see that by inserting dummy qubits among the added qubits we can break the $DT(G)$ in three copies of the same base-graph: one will be used for the computation while the other two can be used %will behave 
as traps. Furthermore for each vertex base-location the choice of where to break the graph is independent from other vertex base-locations and can be made in advance by the client. The server remains totally ignorant of this choice. This choice is called \emph{trap-colouring}.

\begin{definition}[Trap-Colouring, taken from \cite{KW15}]\label{trap colouring} We define trap-colouring to be an assignment of one colour to each of the vertices of the dotted triple-graph that is consistent with the following conditions: 

\begin{enumerate}
\item Primary vertices are coloured in one of the three following colours: white or black (for traps), or green (for computation). %There is an exception for input base-locations (see step 5). 
\item Added vertices are coloured in one of the four following colours: white, black, green or red. 
\item In each primary set $P_v$ there is exactly one vertex of each colour. 
\item Colouring the primary vertices fixes the colours of the added vertices: added vertices that connect primary vertices of different colour are red, added vertices that connect primary vertices of the same colour get that colour.
%\item For input base-locations, instead of green we have a blue vertex (but all other rules, including how it is connected with the other vertices, apply in the same way as if it were green).
\end{enumerate}
\end{definition}

\begin{figure}%[H]
\includegraphics[width=1\columnwidth]{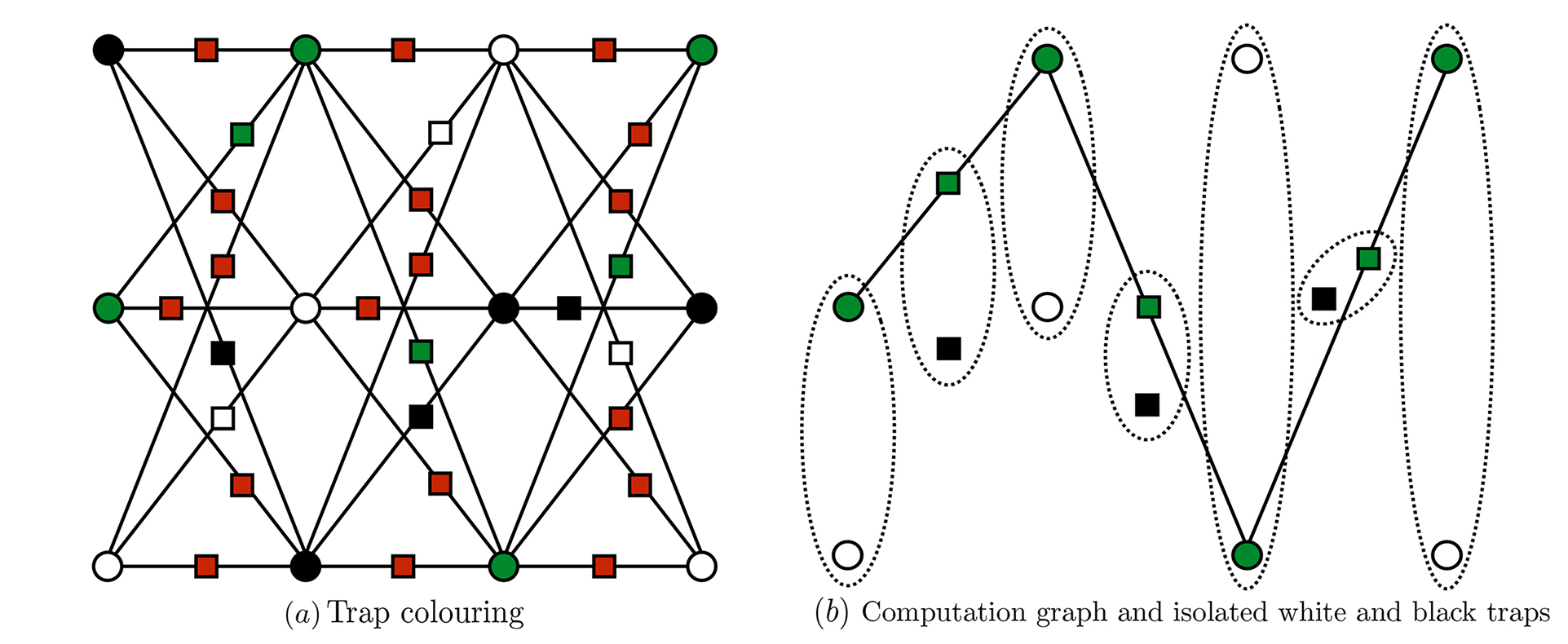}

\caption{Dotted-triple-graph for one-dimensional base graph of four qubits. Circles: primary vertices (base-location : vertex of the base-graph); Squares: added vertices (base-location : edge of the base-graph). (a) Trap-colouring. Green: computation qubits; White/black: trap qubits; Red: dummy qubits. Client chooses the trap-colouring, prepares each qubit individually and sends them one by one for the server to entangle according to the generic construction. (b) After entangling, the breaking operation defined by the dummy qubits will reduce the graph in (a) to the computation graph and for each vertex a corresponding trap/tag qubits.}
\label{figure3}
\end{figure}
The flow, Past of qubit $i$ and Influence-past of qubit $i$ can all be extended to the Dotted-Triple-Graph construction, with the result that each qubit still depends on a constant number of previous measurements (see \cite{KW15}).

For completeness we give the %basic 
verification protocol from \cite{KW15} that we use:

\begin{algorithm}[H]
\caption{Verifiable Universal Blind Quantum Computation using dotted triple-graph (with Fault-tolerant Encoding) - Taken from \cite{KW15}}
\label{prot:KW15}
We assume that a standard labeling of the vertices of the dotted triple-graph $DT(G)$ is known to both the client and the server. The number of qubits is at most $3N(3c+1)$ where $c$ is the maximum degree of the base graph $G$.\\
\noindent$\bullet$ \textbf{Client's resources} \\
-- Client is given a base graph $G$. The corresponding dotted graph state $\ket{D(G)}$ is generated by graph $D(G)$ that is obtained from $G$ by replacing every edge with a new vertex connected to the two vertices originally joined by that edge.\\
-- Client is given an MBQC measurement pattern $\bbbm_{\textrm{Comp}}$ which when applied to the dotted graph state $\ket{D(G)}$ performs the desired computation, in a fault-tolerant way that can detect or correct errors fewer than $\delta/2$.\\
-- Client generates the dotted triple-graph $DT(G)$ and selects a trap-colouring according to definition \ref{trap colouring}, which is done by choosing independently the colours for each set $P_v$.\\
-- Client for all red vertices will send dummy qubits and thus perform break operation.\\
-- Client chooses the green graph to perform the computation.\\ 
-- Client for the white graph will send dummy qubits for all added qubits $a^{e}_w$ and thus generate white isolated qubits %one 
at each primary vertex set $P_{v}$. Similarly for the black graph the client will send dummy qubits for the primary qubits $p^v_b$ and thus generate black isolated qubits %one 
at each added vertex set $A_{e}$.\\
\noindent -- The set $D$ of the positions of dummy qubits is chosen as defined above (fixed by the trap-colouring).\\
\noindent -- A binary string $\mathbf{s}$ of length at most $3N(3c+1)$ represents the measurement outcomes. It is initially set to all zeros.\\
\noindent -- A sequence of measurement angles, $\phi=(\phi_i)_{1 \leq i \leq 3N(3c+1)}$ with $\phi_i \in A = \{0, \pi/4, \cdots, 7\pi/4\}$, consistent with $\bbbm_{\textrm{Comp}}$. We define $\phi_i'(\phi_i,\mathbf{s})$ to be the measurement angle in MBQC, when corrections due to previous measurement outcomes $\mathbf{s}$ are taken into account (the function depends on the specific base-graph and its flow, see e.g. \cite{bfk}). We also set $\phi'_i = 0$ for all the trap and dummy qubits. 
\noindent -- The Client chooses a measurement order on the dotted base-graph $D(G)$ that is consistent with the flow of the computation (this is known to the Server). The measurements within each set $P_v$ and $A_e$ of $DT(G)$ are ordered randomly.
\\
\noindent -- $3N (3c+1)$ random variables $\theta_i$ with values taken uniformly at random from $A$.\\
\noindent -- $3N (3c+1)$ random variables $r_i$ and $|D|$ random variable $d_i$ with values taken uniformly at random from $\{0,1\}$.\\
\noindent -- A fixed function $C(i, \phi_i, \theta_i, r_i, \mathbf{s}) = \phi_i'(\phi_i,\mathbf{s}) +\theta_i + r_i \pi$ that for each non-output qubit $i$ computes the angle of the measurement of qubit $i$ to be sent to the Server.

(Continues on next page)
\end{algorithm}

\begin{algorithm}[H]
\caption{continuing: VUBQC with DT(G)}
\noindent$\bullet$ \textbf{Initial Step}\\
-- \textbf{Client's move:} Client sets all the value in $\mathbf{s}$ to be $0$ and prepares the input qubits:

\AR{
\ket e = X^{x_1} Z(\theta_1) \otimes \ldots \otimes  X^{x_l} Z(\theta_l) \ket I
}
and the remaining qubits in the following form
\AR{
\forall i\in D &\;\;\;& \ket {d_i} \\ 
\forall i \not \in D &\;\;\;& \prod_{j\in N_G(i) \cap D} Z^{d_j}\ket {+_{\theta_i}}
}
and sends the Server all the $3N (3c+1)$ qubits in the order of the labelling of the %vertices of the 
graph.

-- \textbf{Server's move:} Server receives $3N(3c+1)$ single qubits and entangles them according to $DT(G)$.

\noindent$\bullet$ \textbf{Step $i : \; 1 \leq i \leq 3N (3c+1)$}

-- \textbf{Client's move:} Client computes the angle $\delta_i = C(i, \phi_i, \theta_i, r_i, \mathbf{s})$ and sends it to the Server.\\ 
-- \textbf{Server's move:} Server measures qubit $i$ with angle $\delta_i$ and sends the result $b_i$ to the Client.\\
-- \textbf{Client's move:} Client sets the value of $s_i$ in $\mathbf{s}$ to be $b_i + r_i$.\\

\noindent$\bullet$ \textbf{Final Step:}

-- \textbf{Server's move:} Server returns the last layer of qubits (output layer) to the Client.\\

\noindent$\bullet$  \label{step:Alice-prep} \textbf{Verification}\\
-- After obtaining the output qubits from the Server, the Client measures the output trap qubits with angle $\delta_t = \theta_t + r_t \pi$ to obtain $b_t$.

-- Client accepts if $b_i = r_i$ for all the white (primary) and black (added) trap qubits $i$.

-- Client applies corrections according to measurement outcomes $b_i$ and secret parameters $\theta_i, r_i$ at the output layer green qubits and obtains the final output.

\end{algorithm}

\begin{theorem}[Verifiability of VUBQC, taken from \cite{KW15}]\label{vbqc verif}
Protocol \ref{prot:KW15} is $\epsilon_{2}$-verifiable for the Client, where $\epsilon_{2} = \qty(\frac{8}{9})^{d}$ for $d = \ceil*{\frac{\delta}{2(2c+1)}}$. $\epsilon_2$-verifiability means that for any strategy of the server the real state is $\epsilon_2$-close (in trace distance) to $\rho_{ideal}(\rho_{in},p_{ok})$ for some $0\leq p_{ok}\leq 1$ where $\rho_{in}$ is the initial state and $\rho_f:=U(\rho_{in})$ is the desired final state:
%\EQ{\Delta(\tilde\rho(\tilde{P}_{j},\rho_{in}),\rho_{ideal}(\rho'_{in}))&\leq&\epsilon}
%where 

\AR{\rho_{ideal}(\rho_{in},p_{ok})&:=&p_{ok}\rho_f +(1-p_{ok})([\mathsf{abort}_1])%\ket{fail}\bra{fail})
}
%is the state in the ideal execution and $\tilde\rho(\tilde{P}_{j},\rho_{in})$ is the state in the real protocol.
\end{theorem}

\subsection{Standard Definitions of Security for Quantum Two-Party Computations}

Following \cite{DNS10}, we have two parties $A$ and $B$ with quantum registers $\mathcal{A}$ and $\mathcal{B}$ and extra register $\mathcal{R}$, where $dim\mathcal{R} = (dim\mathcal{A} + dim \mathcal{B})$. The input is denoted $\rho_{in} \in D(\mathcal{A}\otimes\mathcal{B}\otimes\mathcal{R})$, where $D(\mathcal{A})$ is the set of all possible quantum states in register $\mathcal{A}$. Let then $L(\mathcal{A})$ be the set of linear mappings from $\mathcal{A}$ to itself and let $\phi : L(\mathcal{A}) \rightarrow L(\mathcal{B})$ be a completely positive and trace preserving superoperator, also called \emph{quantum operation}. Finally let $\bbbone_{\mathcal{A}}$ be the totally mixed state and $\mathbf{1}_\mathcal{A}$ be identity operator in register $\mathcal{A}$. We will write $U \cdot \rho$ instead of $U \rho U^{\dagger}$ and we will sometimes denote $[b] := \dyad{b}$. Let $U_f$ be the unitary which when applied to classical inputs (computational basis) $(x, y)$, returns the classical output $f(x, y) = (f_1(x, y), f_2(x, y))$. The ideal output is $\rho_{out} = (U_f \otimes \bbbone) \cdot \rho_{in}$. Given two states $\rho_{0}$ and $\rho_{1}$, the trace norm distance is $\Delta(\rho_{0}, \rho_{1}):= \frac{1}{2}\norm{\rho_{0} - \rho_{1}}$. When $\Delta(\rho_{0}, \rho_{1}) \leq \epsilon$ then any process applied to $\rho_{0}$ behaves the same as it would on $\rho_{1}$ except with probability at most $\epsilon$. A function $\mu$ is \emph{negligible in $n$} if, for every polynomial $p$, for sufficiently large $n$'s it holds that $\mu(n) < \frac{1}{p(n)}$.

All the proofs of security in this paper will be in the \emph{real/ideal simulation paradigm}: when considering a party as corrupted, we will construct a simulator interacting with the ideal functionality such that they are not able to detect that they are not in fact interacting directly with a real world honest party instead.

The standard definition of security means that the simulated and real states are exponentially close and so indistinguishable for the adversary.

\begin{definition}[Privacy]
We say that the $n$-step two party strategy is $\epsilon$-private for $B$ if there exists $\epsilon(n)$ negligible in $n$ such that for all adversaries $\tilde{\mathcal{A}}$ and for all steps $i$ we have:
\begin{equation}
\Delta\qty(v_{i}(\tilde{\mathcal{A}}, \rho_{in}), \Tr_{\mathcal{B}_{i}}\qty(\tilde{\rho}_{i}(\tilde{A}, \rho_{in}))) \leq \epsilon(n)
\end{equation}
where $v_{i}(\tilde{\mathcal{A}}, \rho_{in})$ is the view of the adversary when interacting with the simulator and $\Tr_{\mathcal{B}_{i}}\qty(\tilde{\rho}_{i}(\tilde{A}, \rho_{in}))$ is the view of the adversary in the real protocol.
\end{definition}

For further details on these definitions, see \cite{DNS10}.% or more precise definitions of protocols, simulator and security in the quantum case.}

Another property which a quantum protocol may satisfy is \emph{verifiability}. This intuitively means that the probability of receiving a corrupted output without aborting is negligible.

\begin{definition}[$\epsilon$-verifiability]
\label{verification}
A protocol is said to be \emph{$\epsilon$-verifiable for party $P_{i}$} if for any (potentially malicious) behaviour of party $P_{j}$ with $j \neq i$, the probability of obtaining a wrong output and not aborting is bounded by $\epsilon$. If the output of the real protocol with malicious party $\tilde{P}_{j}$ is $\tilde{\rho}(\tilde{P}_{j}, \rho_{in})$ then we have that:
\EQ{\label{eq:verification}\Delta(\tilde\rho(\tilde{P}_{j},\rho_{in}),\rho_{ideal}(\rho'_{in}))&\leq&\epsilon}
where 

\AR{\rho_{ideal}(\rho_{in})&:=&p_{ok}(\bbbone_{\H_{P_i}}\otimes \mathcal{C}_{\H_{P_j}})\cdot U_f \cdot(\rho_{in})+(1-p_{ok})([\mathsf{abort}_1])%\ket{fail}\bra{fail})
}
where $\mathcal{C}_{\mathcal{H}_{P_{j}}}$ is the deviation that acts on $P_{j}$'s systems after they receive their outcome (a CP-map, it can be purified by including ancilla), $\rho'_{in} = (\bbbone_{\mathcal{H}_{P_i}} \otimes \mathcal{D}_{\mathcal{H}_{P_{j}}})\rho_{in}$ is an initial state compatible with $P_{i}$'s input, where $\mathcal{D}_{\mathcal{H}_{P_{j}}}$ is a deviation on the input by $P_{j}$.
\end{definition}

\noindent Note that since $\mathcal{C}_{\H_{P_j}}$ is performed at the final step of the protocol, we also have that the global state (before the final step deviation, i.e. at step $n-1$) follows:

\EQ{\label{eq:verification2}\Delta(\tilde\rho^{n-1}(\tilde{P}_{j},\rho_{in}),\rho^{n-1}_{ideal}(\rho'_{in}))&\leq&\epsilon}
where $\rho^{n-1}_{ideal}(\rho_{in}) := p_{ok} U_f \cdot (\rho_{in})+(1-p_{ok})([\mathsf{abort}_1])$%\ket{fail}\bra{fail})$.

During our protocol we will use bit commitment and 1-out-of-2 oblivious transfer (or OT).

Bit commitment consists of two phases, Commit and Reveal, such that after the Commit the receiver has no information about the value that has been committed (hiding), while during the Reveal the sender cannot reveal a value different from the one committed previously (binding). Both these properties can be either computationally or unconditionally verified depending on the scheme (but not both unconditionally). We suppose that all the commitments used verify the following \emph{strict binding} property.

\begin{definition}[Strict binding, taken from \cite{Unr12}]\label{strict bind}
Let $COM$ be a commitment scheme, a deterministic polynomial-time function taking two arguments: the opening information $a$ and the message $y$. We say $COM$ is strictly binding if for all $a, y, a', y'$ with $(a, y) \neq (a', y')$, we have that $COM(a, y) \neq COM(a', y')$. It means that the sender is committed to the (unique) opening information \emph{and} %as well as 
the message that is being committed.
\end{definition}

A 1-out-of-2 oblivious transfer is a two party functionality in which one party ($P_1$ in our case) has two strings $(x_0, x_1)$ and the other ($P_2$) has a bit $b \in \{0, 1\}$. At the end of the protocol $P_2$ recovers $x_b$. $P_1$ should not know which of the strings $P_2$ has chosen while $P_2$ has no information about the string they did not choose $x_{1-b}$.

The following coin-tossing protocol will be needed later: 

\begin{algorithm}[H]
\caption{Coin-Tossing Protocol}\label{coin-toss}
\begin{enumerate}
\item $P_1$ chooses $\alpha_{1} \overset{R}{\in} \{0, 1\}^{\log(s)}$ uniformly at random, commits to it and sends the commitment to $P_2$ (this commitment has to be perfectly hiding).

\item  $P_2$ chooses $\alpha_{2} \overset{R}{\in} \{0, 1\}^{\log(s)}$ uniformly at random and sends it to $P_1$.

\item $P_1$ opens their commitment and reveals $\alpha_1$.

\item They both set $\alpha = \alpha_1 \oplus \alpha_2$, where $\oplus$ corresponds to the bitwise XOR ($\alpha$ is the index of the evaluation graph).

\end{enumerate}
\end{algorithm}

\subsection{Quantum Covert Adversaries}

We now introduce a new adversarial model for quantum protocols, based on the covert adversaries in \cite{AumLin07}. The quantum covert adversaries are also able to deviate arbitrarily from the protocol. The main difference with malicious is that when they cheat they are caught with high probability but not necessarily exponentially close to $1$.

This models real world situations where getting caught might have dire consequences for the parties, eg. financial repercussions. By associating the correct cost to being caught, even if the probability of getting caught is not exponentially close to $1$ the deterrence might be still high enough to make cheating unappealing.

\begin{definition}[Privacy against covert adversaries]
We say that the protocol is $\epsilon$-private for $B$ if for all adversaries $\tilde{\mathcal{A}}$ and for all steps $i$ we have:
\begin{equation}
\Delta\qty(v_{i}(\tilde{\mathcal{A}}, \rho_{in}), \Tr_{\mathcal{B}_{i}}\qty(\tilde{\rho}_{i}(\tilde{A}, \rho_{in}))) \leq \epsilon
\end{equation}
\end{definition}

\noindent Note that we do not have any requirement on $\epsilon$ and so, while we would like it to be close to $0$ it does not necessarily need to be negligible. %anymore.

Stronger and more elegant definitions of covert adversaries that might be adaptable to the quantum case can be found in \cite{AumLin07}, but for reasons specific to our construction (namely the fact that measurement is irreversible and disturbs quantum states) they are not directly applicable here.

\paragraph{Relation with specious and malicious adversaries}
In the classical case the covert adversary lies between the honest-but-curious and malicious ones for some choices of $\epsilon$, as shown in \cite{AumLin07}.
Their analysis partially holds also for the quantum case: the quantum covert adversary is strictly less powerful than the fully malicious adversary, unless the $\epsilon$ is negligible, in which case both notions are trivially equivalent.
The situation is not as clear with quantum specious adversaries. We only get that if the protocol is secure against covert adversaries with a certain $\epsilon$ then it is also secure against specious adversaries with the same $\epsilon$, which is not required to be negligible.

\section{The Quantum Cut-and-Choose Technique}\label{quantum cc}

The Cut-and-Choose method is a standard technique to boost a protocol secure against honest-but-curious to being secure against malicious adversaries, by enforcing the honest-but-curious behaviour.
The classical Yao protocol, which also relies on a client/server (or garbler/evaluator) setup, was first proven secure against honest-but-curious adversaries (e.g. in \cite{LP04proof}) and then against malicious adversaries using this technique in \cite{LinPin07}.

The garbler creates $s$ copies of the graph and the evaluator chooses which ones (the \emph{check graphs}) they will check for consistency. %from now on. 
If the checks pass and additional precautions are taken, the evaluator is confident that with high probability the remaining graphs (the \emph{evaluation graphs}) %they will then use for the computation 
were also constructed correctly and can be used for the computation. Here we will have $s$ graphs in total, $s-1$ \emph{check graphs} and $1$ \emph{evaluation graph}. In the classical Yao protocol, the probability of cheating and not getting caught is made negligible by using $\frac{s}{2}$ check graphs and $\frac{s}{2}$ evaluation graphs and revealing only the majority output of the evaluation graphs (not possible in quantum case, see below). There are several %quite a few 
caveats in this setting even classically, cf. \cite{LinPin07,Kir08}.

We extend this technique to quantum computations in three steps. First, we show how to %to the quantum case for 
verify %ication 
quantum states using Quantum-State-Preparation Cut-and-Choose (QSP-CC). This ensures that the resource state for the quantum computation in VBQC is constructed correctly. Secondly, we define \emph{Classical Instructions Cut-and-Choose} (CI-CC), using the classical Cut-and-Choose to verify that the (classical) instructions for the computation are correct. Finally, we combine the two to get \emph{Quantum Computation Cut-and-Choose} (QC-CC).

\subsection{Quantum State Preparation Cut-and-Choose (QSP-CC)}\label{qscc}

The intuition for this functionality is that it allows the receiver to essentially test that a state $\ket{\psi_\alpha}$ was prepared and sent correctly (as promised), up to a certain probability, without the sender revealing the classical description of that state. This procedure boosts the classical commitment scheme towards a quantum-state commitment. Note however, that a proper quantum-state commitment scheme would require that the receiver also obtains no information about the state, which is not true here: in the above scheme the receivers can always obtain some (partial) information by measuring the state.

See Protocol \ref{qsp-cc} (below) for details and Appendix \ref{proof_cc} for proof %The proof 
that the state is %indeed 
$\frac{1}{\sqrt{s}}$-close to the correctly-prepared. %one can be found in Appendix \ref{proof_cc}.
	Note that if we used more than $1$ evaluation graphs (as needed for boosting the success probability),
the probability of successful cheating does not scale linearly with parallel repetitions of QSP-CC %functionality 
due to coherent (entangled) attacks.

\begin{algorithm}%[H]
\caption{Functionality: QSP-CC}
\label{qsp-cc}
\textbf{Set up:} Two parties Alice and Bob.\\
\textbf{Input:} Alice inputs a set of $s$ pure states $\{\ket{\psi_1},\cdots,\ket{\psi_s}\}$ along with their classical descriptions $\psi_i$. Bob chooses one index $\alpha$ at random.\\
\textbf{Output:} For any strategy of an adversarial Alice, there exists a $0 \leq p \leq 1$ such that Bob obtains a state $\frac{1}{\sqrt{s}}$-close to $\rho_{id}(p) := p\dyad{\psi_\alpha} + (1 - p)[\textrm{Abort}]$.\\
\textbf{Protocol:}\\
-- Alice commits the classical values $\psi_i$ with %using a 
quantum-safe classical commitment scheme.\\
-- Alice sends all the $s$ labelled quantum states and then the $s$ commitments to Bob.\\
%-- Alice then sends all the $s$ commitments to Bob.\\
-- Bob randomly chooses the index $\alpha$ and request %from Alice 
to open all %her classical 
commitments for $i \neq \alpha$.\\
-- Alice reveals all the classical values $\psi_i$ for $i \neq \alpha$.\\
-- Bob measures all states with index $i \neq \alpha$ in the basis $\{P_i := \dyad{\psi_i}, \bar{P}_i := I - \dyad{\psi_i}\}$ and aborts if he obtains the second outcome %(corresponding to $\bar{P}_i$) 
for any measurement. %We assume that this measurement is easy to perform. 
If the states $\ket{\psi_i}$ are tensor products of qubits, as is the case for the states sent by the verifier in VBQC, this measurement can be performed with local single qubit measurements.\\
-- The state $\ket{\psi_\alpha}$ %that is 
(not measured) is guaranteed to be $\frac{1}{\sqrt{s}}$-close to $\rho_{id}(p)$ for some $p$.
\end{algorithm}

\subsection{Classical Instructions Cut-and-Choose (CI-CC)}

To perform the VBQC protocol, even if the resource state is correct, one needs to ensure that the classical instructions, i.e. measurement angles for each qubit $\boldsymbol{\delta}(\boldsymbol{\phi},\mathbf{ran})$, are also correct. The subscripts to bold symbols denote different graphs. %(as opposed to subscripts to plain symbols that denote different qubits of one graph).

The angles $\boldsymbol{\phi}$ are public, %and thus same for each graph, 
the receiver wants to ensure that $(\boldsymbol{\delta}_\alpha,\mathbf{ran}_\alpha)$, for the evaluation graph $\alpha$, is correct (and committed) without learning $\mathbf{ran}_\alpha$. To achieve this, the sender commits to the classical instructions for all graphs $\boldsymbol{\delta}_i$ after sending the (correct) qubits $\ket{\psi_i(\mathbf{ran}_i)}$ to the receiver. When $\mathbf{ran}_i$ is opened, the receiver can deterministically decide if $\boldsymbol{\delta}_i$ is correct (w.r.t. $\boldsymbol{\phi}$). Intuitively we expect that such a classical Cut-and-Choose has a $\frac1s$ probability of failure but it turns out that it is in fact %Nevertheless, counter-intuitively, a classical Cut-and-Choose guarantees that 
$\frac{1}{\sqrt{s}}$, due to the specific proof techniques used to prove security against a quantum adversary and in particular the special rewinding from Appendix \ref{unruh rew}. This will become clear in the proof of the following section.

\begin{algorithm}[H]
\caption{Functionality: CI-CC}
\label{i-cc}
\textbf{Set up:} Two parties Alice and Bob.\\
\textbf{Input:} Alice inputs a set of $s$ pure states $\{\ket{\psi_1(\mathbf{ran}_1)},\cdots,\ket{\psi_s(\mathbf{ran}_s)}\}$ and for each state a set of classical instructions $\{\boldsymbol{\delta}_1,\cdots, \boldsymbol{\delta}_s\}$ 
and the underlying randomness $\{\mathbf{ran}_1,\cdots,\mathbf{ran}_s\}$. Bob chooses one index $\alpha$ at random. 
\\
\textbf{Output:} For any strategy of an adversarial Alice involving only the classical instructions, the probability that Bob receives the wrong instructions $\boldsymbol{\delta}_\alpha$ for $\ket{\psi_\alpha(\mathbf{ran}_\alpha}$ and does not abort is at most $\frac{1}{\sqrt{s}}$.\\
\textbf{Protocol:}\\
-- Alice commits the %classical 
values $\boldsymbol{\delta}_i, \mathbf{ran}_i$ with %using a 
quantum-safe classical commitments.\\ %scheme.\\
-- Alice sends all the $s$ labelled quantum states and then the $s$ commitments to Bob.\\
-- Bob randomly chooses the index $\alpha$ and request %from Alice 
to open all %her classical 
commitments for $i \neq \alpha$.\\
-- Alice reveals all the classical values $\boldsymbol{\delta}_i, \mathbf{ran}_i$ for $i \neq \alpha$.\\
-- Bob %uses the randomness to 
verifies that all the instructions are computed correctly %according to the corresponding state 
and aborts otherwise. %if any of them has not%\footnote{We assume that this verification is polynomial and deterministic and that the knowledge of the state and the randomness is sufficient to do that.}
\\
-- The remaining set of instructions $\boldsymbol{\delta}_\alpha$ is correct up to probability $\frac{1}{\sqrt{s}}$.
\end{algorithm}

The receiver %sender 
can use the remaining committed instructions $\boldsymbol{\delta}_\alpha$ to drive the computation %later on 
by asking the sender to open the instructions %correct ones 
corresponding to the measurement outcomes (influence-past). %) at the appropriate time. 
This differs from the classical case where the circuit evaluation is non-interactive.

\subsection{Quantum Computation Cut-and-Choose (QC-CC)}\label{qc-cc}

We can now introduce a Cut-and-Choose technique for quantum computation. The sender can deviate in any way. By combining CI-CC with QSP-CC and using the commitments during the (quantum) computation, the evaluator knows (with high probability) that they have been asked to perform the correct quantum computation.

\begin{algorithm}[H]
\caption{Functionality: QC-CC}
\label{qc-cc-prot}
\textbf{Set up:} Two parties Alice and Bob, a (quantum) computation $f$ hidden in the pairs $(\mathbf{ran}_i,\boldsymbol{\delta}_i)$.\\ %that both wish to compute, concealed in the hidden random parameters $\mathbf{ran}_i$ and the corresponding instructions $\boldsymbol{\delta}_i$.\\
\textbf{Input:} Alice inputs a set of $s$ pure states $\{\ket{\psi_1(\mathbf{ran}_1)},\cdots,\ket{\psi_s(\mathbf{ran}_s)}\}$ and for each state a set of classical instructions $\{\boldsymbol{\delta}_1,\cdots, \boldsymbol{\delta}_s\}$, the underlying randomness $\{\mathbf{ran}_1,\cdots,\mathbf{ran}_s\}$ and the classical description of each state. Bob chooses one index $\alpha$ at random.\\
\textbf{Output:} For any adversarial Alice, the probability of performing the wrong computation and not aborting is $\order{\frac{1}{\sqrt{s}}}$.
%that Bob does the wrong computation and does not abort is $\order{\frac{1}{\sqrt{s}}}$.\\
\textbf{Protocol:}\\
-- Alice commits the classical values $\psi_i(\boldsymbol{ran}_i)$ and $\boldsymbol{\delta}_i, \mathbf{ran}_i$ as in protocols \ref{qsp-cc} and \ref{i-cc} respectively.\\
-- Alice sends all the $s$ labelled quantum states and then the $s$ commitments to Bob.\\
-- Bob randomly chooses the index $\alpha$ and request %from Alice 
to open all %her classical 
commitments for $i \neq \alpha$.\\
-- Alice reveals all the classical values $\psi_i(\boldsymbol{ran}_i), \boldsymbol{\delta}_i, \mathbf{ran}_i$ for $i \neq \alpha$.\\
-- Bob performs the same verifications as in protocols \ref{qsp-cc} and \ref{i-cc} and aborts similarly.\\% in the same cases.\\
-- Alice reveals a subset of instructions $\boldsymbol{\delta}_\alpha$, as required by the protocol they wish to perform (but keeps secret $\mathbf{ran}_\alpha$).\\
-- Bob uses these along with the state $\ket{\psi_\alpha}$ to perform the desired computation. %protocol and recovers $f_\alpha$.
\end{algorithm}

From Protocol \ref{qsp-cc}, the state $\ket{\psi_\alpha(\mathbf{ran}_\alpha)}$ is $\frac1{\sqrt{s}}$-close to the ideal state. From Protocol \ref{i-cc}, the pair $(\boldsymbol{\delta}_\alpha,\mathbf{ran}_\alpha)$ are constructed correctly up to probability $\frac{1}{\sqrt{s}}$. It follows that the computation is performed correctly up to probability $O(\frac1{\sqrt{s}})$.

\begin{proof}
We suppose that we have access to an oracle $O^f$ which, upon being given the secret parameters $\mathbf{ran}_i, \psi_i(\boldsymbol{ran}_i)$ and a subset of the instructions $\boldsymbol{\delta}_i$ which is sufficient to drive a single specific computation on the resource state $\ket{\psi_i(\mathbf{ran}_i)}$, returns the output $f_i$ produced by using these instructions on the (correct) state $\ket{\psi_i(\mathbf{ran}_i)}$.

Let us construct a simulator for an adversarial Alice:

\begin{enumerate}
\item The simulator runs the protocol normally %as normal for the steps of the protocol 
until Alice reveals her commitments: receives the qubits and commitments, chooses random index $\alpha$ and receives the openings of commitments for $i \neq \alpha$.

\item Using the special rewinding (Appendix \ref{unruh rew}), it rewinds the simulation and %back to before he chose the index and 
chooses at random a second index $\alpha'$,  (for known $\mathbf{ran}_{\alpha'}, \psi_{\alpha'}(\boldsymbol{ran}_{\alpha'})$, $\boldsymbol{\delta}_{\alpha'}$). %This is done using the special rewinding detailed in Appendix \ref{unruh rew}.

\end{enumerate}

This (classical) part of the protocol can be viewed as three steps: commitment (sending the commitments), challenge (choosing the random $\alpha$) and response (revealing the commitments).

Furthermore it has two properties: \emph{special soundness} and \emph{strict soundness}. Intuitively, special soundness means that given two correct communication transcripts with different challenges, an extractor is able to compute a witness (the simulator recovers the secret values). %, which means 
Strict soundness means that given a commitment and challenge, there is a unique acceptable response: here the only accepted response is to decommit the correct committed values (due to perfect and strict binding).

Following the analysis in \cite{Unr12}, protocols possessing such properties are secure against malicious quantum adversaries and can use rewinding against quantum adversaries in their security proof to extract a witness, evading issues naive rewinding faces due to no-cloning. 
This rewinding can only be done once as per the A-style definition of $\Sigma$-protocols in \cite{Unr12}. 
Lemma \ref{q rew state} shows that after this step % of the protocol 
the distance between the real execution %of the protocol 
and the simulation %the way we perform it for this step 
is bounded by $\frac{1}{\sqrt{s}}$.

Importantly, the clear separation between classical information and quantum states in protocols based on \cite{fk} is what makes the rewinding possible %to use this rewinding technique 
on the classical part of the protocol.

\begin{enumerate}[resume]
\item The simulator performs all the checks for $i \neq \alpha'$ and aborts in the same cases as an honest party would.

\item Alice reveals a subset of $\boldsymbol{\delta}_{\alpha'}$ needed to perform the computation. %corresponding to the computation they wish to perform.

\item The simulator can send this subset, along with the corresponding $\mathbf{ran}_{\alpha'}, \psi_{\alpha'}(\boldsymbol{ran}_{\alpha'})$ which were acquired previously, to the oracle $O^f$ and recover the correct output $f_{\alpha'}$, which is returned to Alice at the end of the computation.

\end{enumerate}

%It follows from this simulation that t
\noindent The distance between the ideal and real execution is bounded by $O(\frac1{\sqrt{s}})$. %, which concludes the proof.
\qed

\end{proof}

In addition to %general 
the arguments given in QSP-CC against using $s/2$ evaluation graphs with quantum states, %in the quantum setting, 
%we present here 
we give 
an attack against our protocol showing concretely %specifically 
why this is not applicable: malicious $P_1$ encodes a teleportation of $P_2$'s input in the trap qubits of one graph.
Since the results of all measurements are given back, $P_1$ can recover $P_2$'s input if the corrupt graph is (one of the) evaluated. This attack succeeds with probability $\frac{\textrm{number of evaluation graphs}}{s}$ and so we cannot hope for a better security bound than an inverse polynomial with this version of the protocol. Classically, evaluating an incorrect circuit has minimal influence when using $s/2$ evaluation circuits where only the majority output is returned. Here guaranteeing that the majority of evaluation circuits is correct would not be enough (the evaluation of a circuit gives extra information for the sake of verifiability).

A final crucial observation is that this proof provides an example where proving security against %demonstrates the particularity of trying to prove security properties against 
a quantum adversary is hard, even for a classical functionality (CC). %such as CC. %Notice that 
The part of the protocol that needs rewinding is entirely classical and the same proof (and extra cost) is necessary even for a fully classical CC protocol %ut-and-Choose 
(single evaluation graph). %This comes at a serious cost, as 
The failure probability %security guarantee 
goes from $\frac{1}{s}$ %classically 
to $\frac{1}{\sqrt{s}}$ for quantum adversaries. 
This kind of quadratic gap is unavoidable against quantum adversaries, unless (possibly) one uses totally different proof techniques: %in security when dimensioning their security parameters. It shows that 
it is not sufficient to use cryptographic primitives resistant against quantum computers (eg. based on LWE), 
but proof techniques (and security parameters) should also be modified. 

\section{The 2PQC Protocol}\label{prot}

\paragraph{Ideal functionality.}

\begin{itemize}

\item \textbf{Inputs:} Each party has a classical input, $x, y \in \{0, 1\}^{n}$ for $P_1$ and $P_2$ respectively. The adversary has auxiliary input $z \in \{0, 1\}^{*}$. Honest parties send their input to the trusted party computing $f = (f_1, f_2)$, a corrupted party $i$ either sends $\mathsf{abort}_i$ or any input of length $n$ (computed in poly-time from their input and auxiliary input). 

\item \textbf{Computation by the trusted party:} If the trusted party receives an inconsistent input or $\mathsf{abort}_i$ by any party $i$, it returns $\mathsf{abort}_i$ to both parties. Otherwise the trusted party first sends $f_1(x', y')$ to $P_1$ ($x' = x$ if $P_1$ is honest, similarly for $P_2$), who can then choose to abort if corrupted (by sending $\mathsf{abort}_1$, which the trusted party forwards to $P_2$). If not, then $P_2$ receives $f_{2}(x', y')$ from the trusted party.

\item	\textbf{Outputs:} Honest parties output what they received from the trusted party, corrupted parties have no output, the adversary outputs an arbitrary BQP function of their inputs and outputs.
\end{itemize}

We will use the security parameters $s$ (number of graphs for CC) and $n$ (size of the inputs) throughout the paper. $P_1$ and $P_2$ will denote the garbler/client and the evaluator/server respectively.

\paragraph{High-level overview}
$P_1$ and $P_2$ have already chosen a VBQC graph computing $f = (f_1, f_2)$ fault-tolerantly. $P_1$ chooses and commits to the randomness for the $s$ versions of the graph (the angles $\theta$ of the states and the flips $r$ in the measurements) and also to all the corresponding measurements angles according to the flow, the possible input measurement angles for both parties, the decryption keys for $P_2$'s output also according to the flow and the positions of the traps among $P_2$'s output qubits. For every input bit of $P_2$, they perform a 1-out-of-2 OT at the end of which $P_2$ learns the measurement angles for their input qubits for all graphs %\footnote{
(doing the OTs before sending the qubits is essential for the security proof, it allows the simulator to recover the adversary's input before constructing the graphs). They then perform a QC-CC protocol: the qubits of each graph are the states $\ket{\psi_i(\mathbf{ran}_i)}$, the commitments are $\mathbf{ran}_i, \psi_i(\boldsymbol{ran}_i)$ and $\boldsymbol{\delta}_i$, they choose the evaluation graph with a coin-tossing protocol, $P_1$ reveals the commitments of check circuits and $P_2$ verifies them as well as the states. Then they perform the evaluation with the VUBQC protocol with $P_1$ decommitting to the instructions (measurement angles). At the end they perform a simple key-exchange protocol so that $P_2$ may decrypt their output.
 
\begin{theorem}[Correctness]
If both parties are honest and follow the steps of the protocol then the protocol is correct.
\end{theorem}

\begin{proof}
If the parties are honest, all the graphs and commitments are correct. The protocol (restricted to the evaluation graph) is equivalent to the normal VBQC, with the evaluator ($P_2$) keeping part of the output.  The last step of the protocol allows $P_2$ %him 
to decrypt this output just the same way as the $P_1$ would in a regular execution. Moreover, all the checks pass and there is no abort. The correctness directly follows from the correctness of the VBQC protocol.\qed
\end{proof}

\begin{theorem}\label{thm:qyao_privacy}
Assume that the oblivious transfer protocol is $\epsilon_{2}$-private against malicious adversaries and that the commitments are perfectly hiding and binding. Let $c$ be the maximum degree of the graph, $\delta$ the number of errors tolerated by the fault-tolerant encoding and $s$ the number of graphs constructed as part of the CC. If the protocol is $\epsilon_{2}$-verifiable for $P_1$, then it is $\epsilon_{2}$-private against a malicious $P_2$ and $\epsilon_{1}$-private against a covert $P_1$, where $\epsilon_{2} = \qty(\frac{8}{9})^{d}$ for $d = \ceil*{\frac{\delta}{2(2c+1)}}$ and $\epsilon_{1} = \frac{1}{\sqrt{s}}$
\end{theorem}

%\begin{proof}
\noindent\emph{Proof Sketch.} The proof follows %By combining the results of 
from Lemma \ref{sec P1} ($\epsilon_1$-privacy against a covert $P_1$) and Lemma \ref{sec P2} ($\epsilon_2$-privacy against a malicious $P_2$) (see detailed proof in Appendix \ref{sec proofs}).
%to follow, in which we prove the $\epsilon_1$-privacy against a covert $P_1$ and $\epsilon_2$-privacy against a malicious $P_2$ respectively, Theorem \ref{thm:qyao_privacy} follows. \qed %we obtain directly the result of the theorem.

The simulator for adversarial $P_1$ (Lemma \ref{sec P1}) is very similar to the one in the proof for Protocol \ref{qc-cc}: obtains %he get 
one set of values form a first run (runs as usual until $P_1$ reveals the commitments) then rewinds the adversary to get a second set and recovers the secret parameters of the adversary which then sends to the ideal functionality, thus getting the ideal output. The simulator runs the evaluation graph with a random input, encrypts the ideal output and returns it to the adversary. %See Lemma \ref{sec P1}.

The simulator for adversarial $P_2$ (Lemma \ref{sec P2}) relies on the construction of a graph which has deterministic output (see Lemma \ref{fake graph}). The simulator recovers the adversary's input with the OTs and sends it to the ideal functionality, from which $P_2$'s ideal output is obtained. Then constructs a graph which always produces this output and hides it among the remaining $s-1$ graphs, which are constructed correctly. The simulator biases the choice of evaluation graph with rewinding the coin-toss so that this special graph is chosen. The checks pass and $P_2$ evaluates the fake graph and gets the correct output. %See Lemma \ref{sec P2}.
\qed
%\end{proof}

\begin{algorithm}[H]
\caption{Secure 2PQC - QYao Cut-and-Choose protocol}
\label{qyao cc protool}
\textbf{Input:} $P_1$ has input $x \in \{0, 1\}^{n}$ and $P_2$ has input $y \in \{0, 1\}^{n}$.
\\
\textbf{Auxiliary input:} Functions $f_1$ and $f_2$, a security parameter $s$ (which is a power of $2$) and the description of a fault-tolerant MBQC pattern $C$ such that $C(x, y) = (f_1(x, y), f_2(x, y))$ for classical inputs $(x, y)$.
\\
\textbf{Output:} Party $P_1$ should receive $f_1(x, y)$ and party $P_2$ should receive $f_2(x, y)$.
\\
\textbf{The protocol:}

\begin{enumerate}
\item $P_1$ constructs $s$ copies of the dotted-triple-graph computing $C$ using independent randomness (here they just choose everything in the different graphs, no need to prepare the qubits just yet).

\item $P_1$ constructs commitments to:
\begin{enumerate}
\item All $\theta_{i,q}^{j}$, $r_{i,q}^{j}$ (these correspond to the $\mathbf{ran}_i$, while the $\theta_{i,q}^{j}$ alone correspond to $\ket{\psi_i(\mathbf{ran}_i}$in the QC-CC protocol, all the other commitments correspond to $\boldsymbol{\delta}_i$), $\prescript{\textsc{k}}{}\delta_{i,q}^{j}$ where $i$ runs over all $s$ graphs, $q$ is the index of the base location in the graph, j correspond to the index of computation, trap and dummy qubits for this particular base location (in the order they appear in the triple-dotted-graph $i$, so if the first qubit for this base location is a trap, $\theta_{i,q}^{1}$ will be the angle associated with that trap) for this base location $q$ and $\textsc{k}$ runs over all possible correction values (previous measurements) affecting this position, according to the flow (the position of the traps and dummies in all graphs are given implicitly through these values). They also commit to the values of the dummy qubits.

\item For their and $P_2$'s input they commit to both possible versions of $\prescript{}{b}\delta_{i,q}^{j}$ where $b \in \{0,1\}$, in permuted order for their input and in correct order for $P_2$'s input.

\item They commit to all potential keys (a One-Time-Pad) for each of $P_2$'s output qubits (according to flow).

\item They commit to the positions of the computation qubits, dummies and traps in the last layer of computation for $P_2$'s output qubits.

\end{enumerate}

\item $P_1$ and $P_2$ participate in $n$ instances of a 1-out-of-2 OT protocol, where in each one $P_1$'s inputs are the sets of $\qty{decommit(\prescript{}{0}\delta_{i,q}^{j})}_{0 \leq i \leq s-1}$ and $\qty{decommit(\prescript{}{1}\delta_{i,q}^{j})}_{0 \leq i \leq s-1}$ and $P_2$'s input is %the 
bit $b_{q}$ corresponding to position $q$. In the end of step 3, $P_2$ receives the decommitments to the measurement angles corresponding to all their inputs (%and always 
for the same binary value across all graphs).

\item $P_1$ and $P_2$ perform the QC-CC protocol (\ref{qc-cc-prot}): 
\begin{enumerate}
\item $P_1$ sends the qubits in the states $\ket*{+_{\theta_{i,q}^{j}}}$
\item Then $P_1$ sends the commitments
\item They both pick the random graph index $\alpha$ using the coin-tossing Protocol \ref{coin-toss}
\item $P_1$ opens the commitments from 2.(a), 2.(b), 2.(c), 2.(d) for any graph whose index is not $\alpha$
\item $P_2$ performs checks and outputs $\mathsf{abort}_2$ and halts if any of the checks fail (the checks are the following : the $\delta$s are correctly constructed and are compatible with the choice of $\phi$, $r$ and $\theta$; the traps and dummies are in the correct place (from commitments 2.(a) and 2.(d)); the decryption keys are correct; the values they received for their input via the OTs are consistent (they are in the correct place with regard to their input bit); they verify that all $\ket*{+_{\theta_{i,q}^{j}}}$ are correct by measuring them in the $\theta_{i,q}^{j}$ basis). 
\item Then $P_1$ opens the values from commitments to the $\delta$s in 2.(b) corresponding to their actual binary input for graph $\alpha$. $P_2$ entangles the qubits according to the dotted-triple-graph and evaluates this graph by asking $P_1$ to open the values to $\prescript{\textsc{k}}{}\delta_{i,q}^{j}$ in 2.(a) for \textsc{k} corresponding to the measurements values they obtained on each qubit. If any of the traps are measured incorrectly $P_1$ privately raises a flag \emph{corrupted}. This corresponds to the evaluation phase of Protocol \ref{prot:KW15}.
\end{enumerate}
\item At the end of the computation they perform the following key-release step:

\begin{enumerate}

\item $P_2$ measures all the qubits in the final layer (output qubits) according to the corresponding $\prescript{\textsc{k}}{}\delta_{i,q}^{j}$, as decommitted by $P_1$.

\item They send back all the measurement outcomes corresponding to $P_1$'s output qubits and commit to the ones corresponding to their output qubits.

\item $P_1$ checks all the traps on their qubits and outputs $\mathsf{abort}_1$ if any fail or if the flag \emph{corrupted} was raised during computation, otherwise they decrypt their output using their decryption keys and set the decryption as their output.
(Continues on next page)

\item $P_1$ reveals the positions of traps and dummies in the final layer of the computation by decommitting 2.(d)

\item For these positions $P_2$ reveals the commitments of 5.(b)

\item[(f)] $P_1$ checks that these traps were measured correctly and outputs $\mathsf{abort}_1$ if any fail, otherwise they decommit the decryption keys in 2.(c) corresponding to $P_2$'s output.

\item[(g)] $P_2$ decrypts their output, sets the decryption as their final output and ends the protocol.
\end{enumerate}
\end{enumerate}
\end{algorithm}

\section*{Acknowledgments}
 Funding
from EPSRC grants EP/N003829/1 and EP/M013243/1 is acknowledged.

\appendix

\section{Mapping of a universal set of gate in the MBQC framework}\label{example mbqc}

We present here diagrams taken from \cite{bfk} showing how to translate a universal set of gates to the MBQC model using the brickwork graphs.

\begin{figure}[h]
\includegraphics[width=0.95\columnwidth]{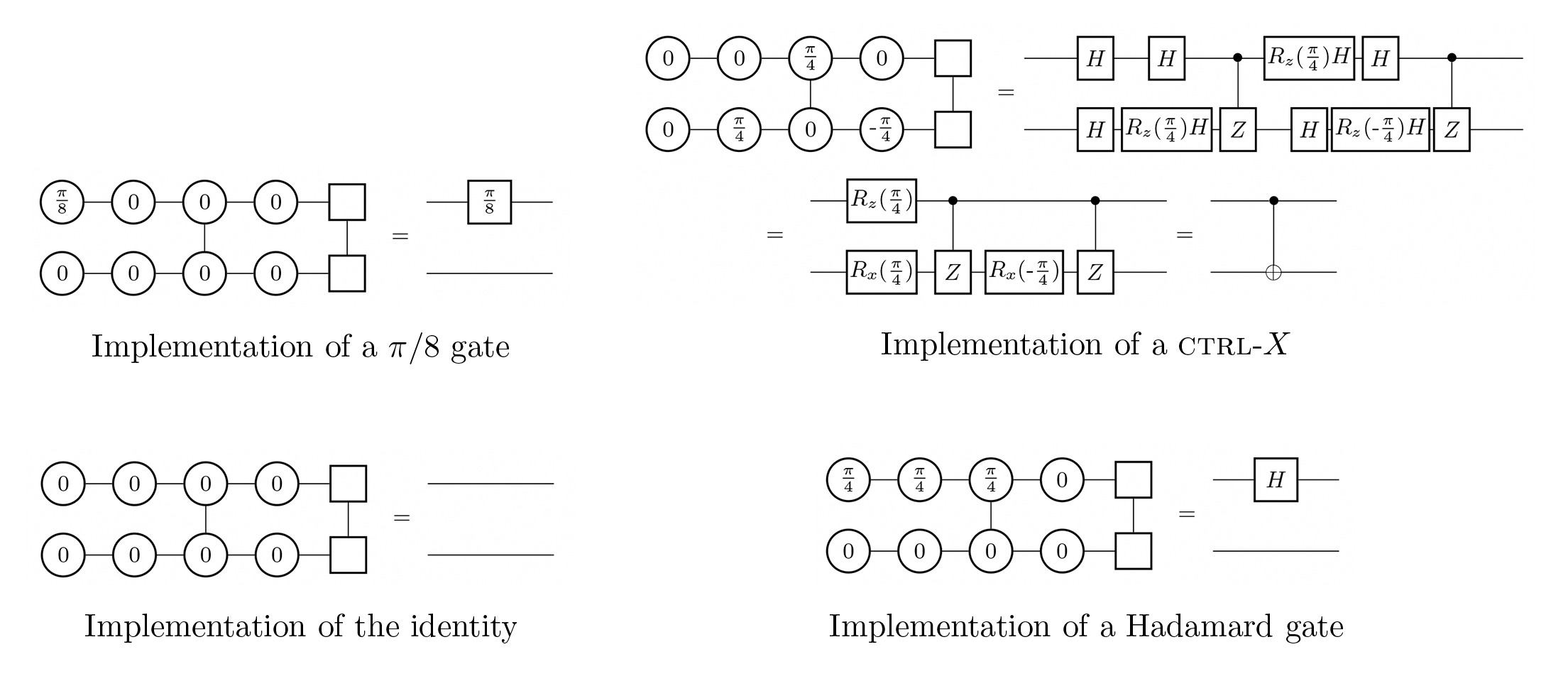}

\caption{Translation of H, $\pi/8$, CNOT and identity gates}
\label{mbqc}
\end{figure}

\section{Proof of QSP-CC}\label{proof_cc}

We have two CP-maps, one corresponding to the real protocol and one to the simulation, and need to show that these maps are $\epsilon$-close. %This will mean that Bob's state is $\epsilon$-close to the simulated one, and the simulated one has trivially the desired property (being equal to $\rho_{id}(p)$).

We define $P_0 := I \otimes_{i=1}^N [\psi_i]$ and $P_j := I \underset{i\neq j}{\bigotimes} [\psi_i] \otimes I_{j}$ for all $1 \leq j \leq s$. It is clear that $P_j P_k = P_0$ if $j \neq k$ and $\comm{P_j}{P_k} = 0$. Moreover we define $P_j = P_0 + P_j^c$ and can easily see that $P_j^c P_k^c = \delta_{j, k}P_j^c$. Also we define $p_2 = \Tr(P_0\rho)$ and $p_1 = \frac{1}{s}\underset{i}{\sum}\Tr(P_i\rho) \leq p_2 + \frac{1}{s}(1 - p_2)$. We have the following two CPTP maps:

\[%\EQ{
\Phi_1(\rho)=\frac{1}{s}\overset{s}{\underset{i = 1}{\sum}} P_i\rho P_i %\expval{\rho}{P_i} 
+ (1 - p_1)[\textrm{Abort}]\]
\[\Phi_2(\rho)=P_0\rho P_0 %\expval{\rho}{P_0} 
+ (1 - p_2)[\textrm{Abort}]\]

\begin{proof}
The real protocol corresponds to Bob acting on some state $\rho$ sent by Alice, by applying the CPTP map $\Phi_1(\cdot)$, i.e. measuring if the states sent are correct by choosing randomly one state to be left unmeasured. We will show that it is $\epsilon$-close to $\Phi_2(\cdot)$ map resulting to $\rho_{id}(p)$. 
Let a purification of $\rho$ be $\ket{\psi}$, we need to show:

\[\Delta(P_0\dyad{\psi}P_0, \frac{1}{s}\underset{i}{\sum}P_i\dyad{\psi}\bra{\psi}P_i) \leq \epsilon\]

\noindent We use the sub-normalised fidelity and its relation with the trace distance. We will use the following properties and definitions from \cite{DFPR13}:

$\tilde{\Delta}\qty(\rho, \sigma) \leq \sqrt{1-\tilde{F}^2(\rho, \sigma)}$

$\tilde{F}(\rho, \sigma) = F(\rho, \sigma) + \sqrt{(1 - \Tr\rho)(1 - \Tr\sigma)}$

$F^2(\ket{\phi}, \sigma) = \expval{\sigma}{\phi}$

\noindent Let $\sigma_1 = \frac{1}{s}\underset{i}{\sum}P_i\dyad{\psi}P_i$, $\sigma_2 = P_0\dyad{\psi}P_0$, $p_2 = \Tr\sigma_2$ and $p_1 = \Tr\sigma_1 \leq p_2 + \frac{1}{s}(1 - p_2)$. It is straightforward to see that $F(\sigma_1, \sigma_2) = p_2$. We obtain: % can now compute:

\[\tilde{F}(\sigma_1, \sigma_2) \geq p_2 + \sqrt{(1 - p_2)(1 - p_2 - \frac{1}{s}(1 - p_2))} \approx p_2 + (1-p_2)(1 - \frac{1}{2s}) \geq 1 - \frac{1}{2s}\]

\noindent assuming $s \gg 1$. It follows that:

\[\tilde{\Delta}(\sigma_1, \sigma_2) \leq \sqrt{1 - (1 - \frac{1}{2s})^2} \approx \sqrt{\frac{1}{s}}\]

\noindent We have $\Phi_1(\rho) = \sigma_1 + (1 - p_1)[\textrm{Abort}]$ and $\Phi_2(\rho) = \sigma_2 + (1 - p_2)[\textrm{Abort}]$ and using the above expressions we get:

\[\Delta(\Phi_1(\rho), \Phi_2(\rho)) \leq \sqrt{\frac{1}{s}}\]

\noindent By tracing out all but the unmeasured system $\alpha$, Bob has (in the simulated case) a state $\frac{1}{\sqrt{s}}$-close to $\rho_{id}(p_2) = p_2[\psi_c] + (1 - p_2)[\textrm{Abort}]$. %and this concludes the proof.
\qed
\end{proof}

\section{Quantum Rewinding}

Classically %In classical cryptography, 
the simulator runs the adversary (chooses its input $x, z, r$) internally and rewinds it by having black box access to the \emph{next message} function  $m_{i+1} = V(x,z, r, m_1, \ldots, m_i)$ where $m_1, \ldots, m_i$ are previous messages. The simulator has to save all messages so that it can send them again later, which is impossible in the quantum setting (due to no-cloning). We present two techniques given in \cite{Wat09} and \cite{Unr12} which achieve a similar result, with different constraints, show that they are applicable for the simulators of our protocol and calculate the success probability for both cases. 

\subsection{Watrous' Oblivious Quantum Rewinding}\label{wat rew}

Let $Q$ be a unitary acting on the pure state $\ket{\psi}\ket{0^{k}}$. We first apply $Q$ to $\ket{\psi}\ket{0^{k}}$ and then measure the first qubit in the computational basis. Let $p(\psi) \in (0, 1)$ be the probability that this measurement outcome is $0$. Then there are unique unit vectors $\ket{\phi_{0}(\psi)}$ and $\ket{\phi_{1}(\psi)}$ such that:
\begin{equation}\label{Eq:Watrous}
Q\ket{\psi}\ket{0^{k}} = \sqrt{p(\psi)}\ket{0}\ket{\phi_{0}(\psi)} + \sqrt{1 - p(\psi)}\ket{1}\ket{\phi_{1}(\psi)}
\end{equation}
Lemma 8 from \cite{Wat09} gives a procedure constructing from $Q$ outputting a state close to $\ket{\phi_{0}(\psi)}$ for any $\ket{\psi}$.

Q represents an attempt at simulating for some cheating adversary, $\ket{\psi}$ is the internal state of this adversary. Getting $0$ means the simulation was successful and getting $1$ corresponds to failure and necessity to rewind. Lemma 8 from \cite{Wat09} states that this is possible if $p(\psi)$ is non-negligible and independent of $\psi$. Rewinding gives a state $\epsilon$-close to that of a successful simulation for any exponentially small $\epsilon$ with polynomially many rewinds.

In our case the rewinding takes place during the coin-tossing part of simulation for $P_2$. Let $\hat{\alpha}$ be the random evaluation graph index chosen by the simulator at the beginning of the proof. After the coin-tossing phase of the protocol, before verifying if the simulation has succeeded or not, the state of the system is in product form (i.e. the random choices of $\alpha, \hat{\alpha}$, are depicted as an equal superposition, but are totally uncorrelated):

\[\ket{\Phi_f}:=\qty(\overset{s}{\underset{\alpha = 1}{\sum}}c_{\alpha}\ket{\phi(\psi, \alpha)}\ket{\alpha})\qty(\overset{s}{\underset{\hat{\alpha} = 1}{\sum}}\frac1{\sqrt{s}}\ket{\hat{\alpha}})\]

\noindent where $\underset{\alpha}{\sum} \abs{c_\alpha}^2 = 1$ are coefficients, $\ket{\phi(\psi, \alpha)}$ is the (normalised) state at the end of the protocol given initial state $\psi$ and choice of graph $\alpha$, while the last part ($\ket{\hat{\alpha}}$) corresponds to the random choice of graph made by the simulator. The projection to the subspace that does not need rewinding (where $\alpha = \hat{\alpha}$) is given: $P_0 := \overset{s}{\underset{\alpha' = 1}{\sum}} I \otimes \dyad{\alpha'} \otimes \dyad{\alpha'}$. We have:

\[\ket{\Phi_f} = P_0\ket{\Phi_f} + (I - P_0) \ket{\Phi_f}\]

\noindent To bring it in the form of Eq.(\ref{Eq:Watrous}) we %first need to 
rewrite it. We define the (normalised) states:

$\ket{\phi_0(\psi)}=\overset{s}{\underset{\alpha = 1}{\sum}} c_{\alpha}\ket{\phi(\psi, \alpha)}\ket{\alpha}\ket{\alpha};\quad 
\ket{\phi_1(\psi)}=\overset{s}{\underset{\alpha = 1}{\sum}} \underset{\alpha \neq \hat{\alpha}}{\sum} \frac{c_{\alpha}}{\sqrt{s - 1}}\ket{\phi(\psi, \alpha)} \ket{\alpha} \ket{\hat{\alpha}}$

\noindent Then we have:

\[\ket{\Phi_f} = \sqrt{\frac{1}{s}}\ket{\phi_0(\psi)} + \sqrt{1 - \frac{1}{s}}\ket{\phi_1(\psi)}\]

\noindent Now, following the unitary action that led to $\ket{\Phi_f}$, we perform the measurement $\{P_0, I - P_0\}$ and store the outcome in the value of an extra qubit (the first one):

\[\ket{\Phi_f} = \sqrt{\frac{1}{s}}\ket{0}\ket{\phi_0(\psi)} + \sqrt{1 - \frac{1}{s}}\ket{1}\ket{\phi_1(\psi)}\]

\noindent Finally, we note that this is exactly in the form of Eq. (\ref{Eq:Watrous}) where $p(\psi) = \frac{1}{s}$ is constant (independent of $\psi$). Further details in \cite{Wat09}, proof of Lemma 8.

\subsection{Unruh's Special Quantum Rewinding} %(taken from \cite{Unr12})}
\label{unruh rew}

Watrous' lemma only ensures that the simulation is successful, but \emph{no information} is kept between two rewinds (hence \emph{oblivious rewinding}).
In the simulator %When simulating 
for covert $P_1$ we need two transcripts in order to recover their input (which is otherwise secret), so another type of rewinding is necessary. We present the conditions under which Unruh's rewinding method \cite{Unr12} applies %can be used 
and then %proceed to 
show that these conditions are met in our protocol and calculate an upper bound on the distance between the actual and simulated (rewinding to extract input) runs. 

Let $\Pi$ be a protocol between $P_1$, with input $(x, w)$, and $P_2$, with input $x$ and output in $\{0, 1\}$, with three messages: commitment $com$ by $P_1$, challenge $ch$ sampled (efficiently) 
uniformly at random by $P_2$ from the set $C_x$ (membership in $C_x$ has to be %should be 
easy to decide), and response $resp$ by $P_1$. $P_2$ accepts (outputs $0$) by a deterministic poly-time computation on $(x, com, ch, resp)$ (it is called an \emph{accepting conversation} for $x$).

\begin{definition}[Special soundness]
Such a protocol has \emph{special soundness} if there is a deterministic poly-time algorithm $\mathsf{K}_{0}$ (the \emph{special extractor}) such that for any two accepting conversations $(com, ch, resp)$ and $(com, ch', resp')$ for $x$ with $ch \neq ch'$, we have that $w := \mathsf{K}_{0}(x, com, ch, resp, ch', resp')$.
\end{definition}

\begin{definition}[Strict soundness]
%We say 
Such a protocol has \emph{strict soundness} if for any two accepting conversations $(com, ch, resp);(com, ch, resp')$ for $x$, %we have that 
$resp = resp'$.
\end{definition}

\noindent\emph{Canonical extractor (\cite{Unr12})}
The extractor runs the first step of the adversary %in order 
to recover $com$ (each step is a separate unitary operation), chooses two values $ch, ch' \in C_{x}$, runs the second step with $ch$ to get $resp$, applies the inverse of the second step (this is the rewinding) and reruns the second step with $ch'$ to get $resp'$ before applying $\mathsf{K}_{0}$. Since each response is uniquely determined by the commitment and the challenge, if the measured response of the adversary is correct then they must have sent a state close to the real response, therefore this does not disturb too much the internal state of the adversary. %This idea is formalised with a game-based proof.

In our case, $com$ corresponds to step 5 (sending commitments), $ch$ is step 6 (result of coin-toss) and $resp$ is step 7 (revealing commitments). We calculate the distance between the internal state of the adversary in the real protocol where only the commitments for the check graphs are opened and the one after the rewinding in the simulation where all commitments have been opened.

We consider the opening of the commitments as a measurement performed on a state shared with the adversary. Let $P_i$ denote the projector corresponding to one run of this part of the protocol where the challenge is $i$ (analogous to the $P_{ch}^{*}$ in %the proof of 
\cite{Unr12}). This corresponds to revealing all commitments but one, corresponding to the graph of index $i$. We define $\bar{P}_i = \bbbone %\mathbb{I} 
- P_i$. There is a ``correct'' subspace $P$, that projects to the subspace that answer all $s$ tests, i.e. $PP_i := P$ for all $i$.

The strict soundness property means that there is a unique classical response to each challenge. This intuitively means that two projection $P_i$ and $P_j$ acting on the same subsystem, are either identity or identical which essentially means that they commute, i.e. $P_i P_j = P_j P_i$. On the other hand, the property of special soundness means that any two of these tests, when both successful, allow the simulator to recover a “witness”. This means that $P_i P_j = P$ for all $i \neq j$.

\begin{lemma}\label{q rew state}
Let $\ket{\psi}$ be any state and let $\qty{P_i}_{0 \leq i \leq s}$ and $P$ be as defined above. For $\epsilon = \frac{1}{s}$ we have that:

\[\max_{\ket{\psi}}\Delta\qty(P\ket{\psi},\frac{1}{s}\underset{i}{\sum}P_i\dyad{\psi}P_i) \leq \sqrt{\epsilon}\]
%where $\epsilon = \frac{1}{s}$
\end{lemma}

\begin{proof}
The proof follows the one above for QSP-CC. We again use the sub-normalised fidelity. Let $\rho = P\dyad{\psi}P$ and $\sigma = \frac{1}{s}\underset{i}{\sum}P_i\dyad{\psi}P_i$. We further define $P_i = P + P_i^c$ and $p = \expval{P}{\psi}$. It follows directly from special soundness that $P_i^c P_j^c = \delta_{i, j} P_i^c$. Since $P_i^c, P_j^c, P$ are orthogonal and the sum of the traces (probabilities) of the $P_i^c$-terms cannot exceed $(1 - p)$, we then have that:

\[\Tr\sigma := \Tr (\frac{1}{s}\underset{i}{\sum}P_i\dyad{\psi}P_i) = \frac{1}{s}\underset{i}{\sum}\expval{(P + P_i^c)}{\psi} = p + \frac{1}{s}\underset{i}{\sum}\expval{P_i^c}{\psi} \leq p + \frac{1}{s}(1 - p)\]

$F^2\qty(P\dyad{\psi}P, \frac{1}{s}\underset{i}{\sum}P_i\dyad{\psi}P_i) = p^2$

\noindent Assuming $\epsilon = \frac{1}{s}$ where $s \gg 1$, with simple calculation we obtain: %can now compute:

\[\tilde{F}(\rho, \sigma) \geq %p + \sqrt{(1 - p)(1 - p - \frac{1}{s}(1 - p))} = p + (1 - p)(1 - \frac{\epsilon}{2} + o(\epsilon)) \geq 
1 - \frac{\epsilon}{2}\Rightarrow\quad \tilde{\Delta}\qty(\rho, \sigma) \leq \sqrt{1 - (1 - \frac{\epsilon}{2})^2} \approx \sqrt{\epsilon}\]

\noindent This result is independent of $\ket{\psi}$ (and $p$) and thus completes the proof.
\qed
\end{proof}

We have $\Phi_1(\cdot)\stackrel{\sqrt{\epsilon}} \approx \Phi_2(\cdot)$ where $\Phi_1(\cdot)$ is the CP-map corresponding to the projection in the $P$ subspace, while $\Phi_2(\cdot)$ is the CP-map corresponding to the real protocol operation (project in one of $P_i$ subspace, randomly chosen from the $s$ possible challenges).
%The simulator, instead of following the real protocol, projects in the subspace $P$, by making a measurement $\{P, \bbbone %\mathbb{I} 
%- P\}$.
The simulated view is $\sqrt{\epsilon}$-close to the real protocol. Given this measurement, the simulator, (a) if it accepts (%i.e. the result of the 
measurement result $P$) %then 
the state is close to the real protocol, (b) if it rejects, %then 
it is identical with the real protocol (since for %in the 
abort, %case, 
the exact state of the parties is irrelevant).

\section{Proof of Security}\label{sec proofs}

All the following proofs will be carried out in the OT-hybrid model. This means that parties can at the same time communicate with one another but also rely in certain steps of the protocol on an ideal call to a trusted party (also called oracle) performing an oblivious transfer. During the simulation, the simulator replaces this trusted party and receives all inputs that the adversary sends to it. We make this assumption only in order to make the proof clearer and easier, in a real protocol this ideal functionality can then be replaced with any OT protocol which is secure against malicious quantum adversaries.

We start by proving the $\epsilon$-verifiability for the client.

\begin{lemma}[$\epsilon$-verifiability for the client]
Protocol \ref{qyao cc protool} is $\epsilon_{2}$-verifiable for $P_1$, where $\epsilon_{2} = \qty(\frac{8}{9})^{d}$ for $d = \ceil*{\frac{\delta}{2(2c+1)}}$.
\end{lemma}

\begin{proof}%[Sketch proof.]
In this proof we consider that $P_1$ is honest while $P_2$ is malicious. During the first steps of the protocol, $P_2$ only participates in the perfectly secure OT at the end of which they receives their input measurement angles, based on their choice of deviated input $\hat{x}$.

$P_2$ then receives the qubits of all the graphs and the commitments. These are perfectly hiding so no information can be recovered from them until $P_1$ decides to send the decommitment values. Furthermore, deviating at this point on the qubits is equivalent to deviating later.

Because the coin-tossing protocol used is proven secure (the same rewinding techniques as the ones used in the proofs for $P_1$ and $P_2$ can be used), a malicious $P_2$ cannot bias the result. After the result of the coin-toss is known, $P_2$ receives the decommitments for all graphs but the evaluation graph, about which they therefore remain blind.

The evaluation part of the protocol follows exactly the same pattern as the VUBQC protocol in \cite{KW15}, with the difference that instead of announcing the angles of the measurements, $P_1$ sends the corresponding decommitments, which in the case of an honest $P_1$ is perfectly equivalent. This protocol has been proven to be $\epsilon_{2}$-verifiable for $P_1$ in \cite{KW15}, see Theorem \ref{vbqc verif}.

After the computation $P_2$ measures the output qubits and has to commit to the result of the measurements, of which they have to reveal the traps. $P_1$ can therefore verify that all the traps were measured correctly. The only last deviation the server is allowed to perform is a further computation on their binary output, which is by definition allowed even in the ideal case.

From this analysis it follows that the exact same verification properties from the protocol in \cite{KW15} hold for this protocol, namely that our protocol is $\epsilon_{2}$-verifiable for $P_1$, where $\epsilon_{2} = \qty(\frac{8}{9})^{d}$ for $d = \ceil*{\frac{\delta}{2(2c+1)}}$, which completes the proof.\qed
\end{proof}

The honest behaviour of $P_2$ is essentially enforced by the traps (verifiability), while the Cut-and-Choose technique enforces $P_1$ to also behave honestly.

\subsection{Security Against Covert $P_1$}\label{secP1}

\paragraph{Intuition} The proof is an adapted version of Lindell and Pinkas’ proof for malicious $P_1$ \cite{LinPin07} and the proof for covert adversaries of \cite{AumLin07}. The intuition %behind the security against covert $P_1$ 
is that if they choose to cheat by constructing incorrectly the graphs or the OTs then they will get caught with probability at least $\epsilon_{1}$ during the opening phase of the protocol, whereas if they do construct the graphs and OTs correctly then they acts the same as an honest party would. %and so nothing is leaked. 
%We will construct a 
The simulator %that 
runs the protocol normally until the opening of the commitments by $P_1$. Then, since the coin toss is random even with malicious $P_1$, the simulator does quantum rewinding (see \cite{Unr12}) and learn the rest of the commitments. When the malicious $P_1$ %later 
reveals their input measurement angles, having obtained the related %by knowing the 
randomness %behind its construction 
the simulator can deduce %from them 
$P_1$'s %the malicious party's 
bit-input and send it to the ideal functionality. It then runs the evaluation graph with random input and returns %sends back 
the output %it 
received from the ideal functionality (note that %after applying 
the correct decoding key is known to %to it (which once again 
the simulator from the rewinding).

\begin{lemma}\label{sec P1}
Protocol \ref{qyao cc protool} is $\epsilon_{1}$-private against a covert $P_1$, where $\epsilon_{1} = \frac{1}{\sqrt{s}}$
\end{lemma}

\begin{proof}
Let us construct the simulator in the following way:

\begin{enumerate}

\item The simulator runs the protocol normally until step 3 :  participates in the OTs, with a random input $y'$ %chosen by it 
instead of $P_2$'s actual input.

\item Then acts %the same way 
as the simulator %against adversarial Alice in %the proof of 
in Protocol \ref{qc-cc-prot}: receives the qubits and commitments, participates in the coin-tossing protocol, receives the openings, rewinds, performs the coin-tossing again and receives the other set of openings, performs all the verifications and aborts as %in all the cases where 
an honest party would.

\item The adversary decommits the value of their input angles for the new evaluation graph, the simulator deduces $P_1$'s binary deviated input $\hat{x}$ from the knowledge of secret parameters and sends it to the ideal functionality. It gets in return the output value $f_{1}(\hat{x}, y)$ for this deviated input.

\item Then performs the computation on graph $\alpha'$ using a random input $y'$. At the end of the computation replaces the bits of the %it would have sent as the 
computation positions for $P_1$'s output with the bits received from the ideal functionality (after correcting them %by flipping some, according to 
using the decryption keys %which it knows because of 
known due to the rewinding) and returns %all 
this to the adversary, while committing to what it computed as $P_2$'s output. When $P_1$ reveals the positions of the traps and dummies in the last layer, the simulator verifies that they were traps as an honest $P_2$ and opens the corresponding positions. Receives the keys from the adversary, and outputs whatever the adversary would %does 
and halts.

\end{enumerate}

\noindent The fact that the simulation is the same up to a negligible factor whether we use a random index chosen by $P_2$ or a coin-tossing protocol can be formalised in terms of a series of games, involving an ideal functionality for coin-tossing (such an ideal functionality takes as input a dummy input $\lambda$ from $P_2$ and returns to both players the same random string $\alpha$):

\noindent\textbf{Game 1.} Here the simulator runs the protocol as usual, with no rewinding.

\noindent\textbf{Game 2.} We replace the protocol for coin-tossing with the ideal functionality. Because the protocol for coin-tossing is secure against malicious quantum adversaries, the distance between the first game and the second is negligible.

\noindent\textbf{Game 3.} Here the simulator sends directly to covert $P_1$ a challenge chosen at random instead of calling the ideal functionality for coin-tossing. This is equivalent to the setup $(com, ch, resp)$ for \cite{Unr12}. This is indistinguishable from the previous game from the point of view of the adversary.

\noindent\textbf{Game 4.} Now we use quantum rewinding so the simulator sends two challenges to $P_1$, this is equivalent to the rewinding performed in \cite{Unr12}. According to the analysis performed in Lemma \ref{q rew state}, here the distance is $\frac{1}{\sqrt{s}}$.

\noindent\textbf{Game 5.} We perform the switch the other way : we replace the simulator sending the random challenges by two calls to the ideal functionality (there is no problem with calling this ideal functionality twice as the input is a dummy input). This is once more indistinguishable for the adversary.

\noindent\textbf{Game 6.} Again we switch back : the ideal functionality calls are replaced by the real coin-tossing protocol. The distance is also negligible and this game represents exactly what happens during our simulation.

The result of this game-based analysis is that the distance between the real execution of the protocol and the simulation the way we perform it for this step is bounded by $\frac{1}{\sqrt{s}}$, up to a negligible quantity.

Following the analysis in Appendix \ref{unruh rew} and Lemma \ref{q rew state}, we get that after the rewinding we have (with $(\hat{x}, y)$ being the input after the adversary's deviation, equivalent to the $\rho_{in}'$ from the definition of verifiability):

\[\Delta\qty(v_{i}(\tilde{\mathcal{P}}_1, \hat{x}), \Tr_{\mathcal{P}_{2, i}}\qty(\tilde{\rho}(\tilde{A}, \hat{x}))) \leq \frac{1}{\sqrt{s}}\]

%Therefore t
The state at the end of the simulation is %also 
$\epsilon_1=\frac{1}{\sqrt{s}}$-close to the real execution. 
\qed
\end{proof}

\subsection{Security Against Malicious $P_2$}\label{secP2}

We first prove a lemma, showing how to construct a graph evaluating to a given fixed output for any input in a way %which is 
indistinguishable from a regular computation %which would have 
producing that output for a given (possibly different) input.

\begin{lemma}\label{fake graph}
Given the value $f_2(x, \hat{y})$, there exist a dotted-triple graph, along with commitments identical %such as those corresponding 
to the evaluation graph, which when evaluated %at 
for any classical input $x'$ and $y'$ returns the fixed output $f_2(x, \hat{y})$ and is indistinguishable for the %from the point of view of the 
evaluator (even when given the decommitments %in the same way as 
for the evaluation graph in Protocol \ref{qyao cc protool}) from a %real 
DT(G) %dotted-triple graph 
computing $(f_1(x, y), f_2(x, y))$ for any $x$ and $y$.
\end{lemma}

\begin{proof}
The graph is constructed as such:
\begin{itemize}

\item $P_1$ chooses all the parameters ($\theta_{q}^{j}$, $r_{q}$, $\prescript{\textsc{k}}{}\delta_{q}^{j}$,  and $\prescript{\textsc{}}{b}\delta_{i,q}^{j}$  for inputs) at random for all the base-location qubits apart from those corresponding to $P_2$'s output.

\item For those base locations $q$, for the computation qubits $j$, $P_1$ chooses at random $\theta_{q}^{j}$ and $b_{q} \in \{0, 1\}$ and sets $\prescript{\textsc{k}}{}\delta_{q}^{j} = \theta_{q}^{j} + b_{q}\pi$ for all choices of \textsc{k}. $P_1$ then sets the corresponding decryption keys to $\prescript{\textsc{k}}{}k_{q} = f_{2}(x, \hat{y})_{q} \oplus b_{q}$ for all choices of \textsc{k}.

\item Then $P_1$ prepares all qubits in the correct state for the base locations but instead of sending the correct ones for the “added” qubits for the edges linking $P_2$'s outputs qubits to the rest of the graph, $P_1$ sends dummy qubits (in states chosen from $\{\ket{0}, \ket{1}\}$). This will have the effect of breaking away these qubits from the rest of the graph.

\end{itemize}

Since the dummies isolate the output of $P_2$ from the rest of the graph, $P_2$ obtains %in the end get 
the result $b_{q}$ when measuring the last layer whatever the previous measurement outcomes are. Thus %and so after 
applying the decryption keys receives $f_{2}(x, \hat{y})$. The indistinguishability follows directly from the server’s blindness in the VBQC protocol and the hiding property of the bit commitment scheme:  since the random $\theta_{q}^{j}$s are not revealed at any point server %they 
cannot get any information from the decommitments if %they are 
done %the same way 
as in %for 
the evaluation graph in Protocol \ref{qyao cc protool}.\qed
\end{proof}

\paragraph{Intuition.} The following proof combines the proofs for a malicious $P_2$ from \cite{LinPin07} and \cite{AumLin07}. The simulator first extracts $P_2$'s deviated input $\hat{y}$ (once again $(x, \hat{y})$ plays the same role as the $\rho_{in}'$ in the verifiability definition) from the OT-protocols, which it then uses to call the ideal functionality and receive $f_2(x, \hat{y})$. Then it chooses at random one graph index for which it will construct the associated graph such that it computes always this value as $P_2$'s output (as in Lemma \ref{fake graph}). %the proof of the previous lemma). 
Then uses the rewinding technique from \cite{Wat09} to bias the coin-toss so this index is picked as the evaluation graph. The other graphs are constructed correctly and all checks pass. Then for the evaluation it follows the same steps as in original protocol, guaranteeing that $P_2$ will receive $f_2(x, \hat{y})$ at the end because the computation over a single graph is verifiable and secure against a malicious $P_2$. While the property of $\epsilon_2$-verifiability for $P_1$ guarantees %assures us 
that the protocol is close to the ideal, %execution, 
the previous lemma gives the %explicit 
construction %followed by 
of a simulator %in order to simulate 
of this ideal execution. Here we justify the coin tossing protocol, since if $P_2$ chose the evaluation graph we would be unable to obtain the simulator for $P_2$.

\begin{lemma}\label{sec P2}
If Protocol \ref{qyao cc protool} is $\epsilon_{2}$-verifiable for the client, then it is $\epsilon_{2}$-private against a malicious server.
\end{lemma}

\begin{proof}
Let us construct the simulator in the following way:
\begin{enumerate}

\item The simulator chooses at random the values of $\prescript{\textsc{}}{b}\delta_{i,q}^{j}$ for $P_2$'s input and organizes these values in $2n$ sets of length $s$: $\qty{\prescript{}{0}\delta_{i,q}^{j}}_{0 \leq i \leq s-1}$ and $\qty{\prescript{}{1}\delta_{i,q}^{j}}_{0 \leq i \leq s-1}$ for all qubits $q$ corresponding to $P_2$'s input.

\item Then invokes the adversary and receives from them what it would have sent to the trusted party computing the OT: $\hat{y} = (\hat{y}_{1}, \ldots, \hat{y}_{m})$ which is the actual input that the adversary intended to use. The simulator returns %sends back 
the corresponding sets of inputs. Here the simulated and real state %and the real state 
are identical. %indistinguishable.

\item The simulator calls the ideal functionality and receives the value $f_{2}(x, \hat{y})$. Chooses the index of the evaluation graph $\alpha$ and for all the check graphs constructs them normally (only now the values of $\theta_{i,q}^{j}$ and $r_{i,q}^{j}$ are determined by $\prescript{}{b}\delta_{i,q}^{j}$ and $\phi_{b}$) and for the evaluation constructs it as in the proof of Lemma \ref{fake graph}. Commits to all those values as it would in the real protocol, in the same order, as well as the decryption keys from Lemma \ref{fake graph}. %the construction in the previous lemma. 
The simulator sends the qubits and gives all these commitments to $P_2$. %the adversary. 
The ideal and real situations are again indistinguishable as shown in Lemma \ref{fake graph}. %by the previous lemma.

\item The simulator needs to ``bias'' %fake 
the coin toss to %such that it 
outputs the $\alpha$ chosen: %which it chose before:
\begin{itemize}

\item It first generates a perfectly binding commitment $\hat{c}$ to a random value $\alpha_{1}$ and sends it to the adversary.

\item Then receives $\alpha_{2}$ from the adversary.

\item If $\alpha_{1} \oplus \alpha_{2} = \alpha$, the simulator continues by decommiting $\alpha_{1}$. Otherwise rewinds back to the beginning of this step, %starting over 
with fresh randomness and a new value for $\alpha_{1}$.

\end{itemize}
\end{enumerate}

\noindent It should be noted that \emph{no information is kept between rewinds} as it is only used to ensure that the adversary is forced to pick our fake graph as the evaluation graph. %Also of importance is the fact that 
The probability of success of the rewinding is $\frac{1}{s}$, which is independent of the initial state of the adversary. We can therefore use the oblivious quantum rewind technique from \cite{Wat09} for this step of simulation. At this point $\mel{\phi_{0}(\psi)}{\rho(\psi)}{\phi_{0}(\psi)} \geq 1 - \epsilon$, where $\ket{\psi}$ is the state of the adversary before this step, $\ket{\phi_{0}(\psi)}$ corresponds to the state after a success happening in one try while $\rho(\psi)$ is the state at the end of the rewinding process. The expected number of rewinds is $\order{\frac{s^2}{s - 1} log{\frac{1}{\epsilon}}}$. Here we want an $\epsilon$ negligible in $n$, so $log{\frac{1}{\epsilon}} = \order{n}$ is sufficient and we get $\order{ns}$ rewinds. More details %about this part of the proof 
can be found in Appendix \ref{wat rew}.

\begin{enumerate}[resume]

\item %After this t
The simulator opens all commitments as in the protocol for the checks. These %checks 
correspond to the correct graphs %that were 
and thus %therefore they 
all pass. %They give n
Nothing is leaked about the %remaining 
evaluation graph and so the ideal and real states are indistinguishable.

\item The simulator then runs %proceeds with the evaluation of 
the remaining graph and %, during which it 
runs the computation until the end of the protocol on the modified evaluation graph as %would an honest 
$P_1$. In the end, outputs whatever the adversary outputs and halts.

\end{enumerate}
Let $k^{P_2}$ be the key committed as the decryption key for $P_2$'s output as part of the fake graph construction in Lemma \ref{fake graph}. At the end of the computation, $P_2$ is in possession of the state $\Tr_{\H_{P_1}}\qty(E_{k^{P_2}}(f(x, \hat{y}))) = E_{k^{P_2}}(f_{2}(x, \hat{y}))$, where $E_{k^{P_2}}(z)$ denotes the encryption as a quantum state of the classical value $z$ under the key $k^{P_2}$ (this key is the one chosen as part of the construction of the fake graph in the previous lemma). At the end of the computation the simulator reveals the key by decommitting it and $P_2$ can decrypt their outcome. The server's view of the real protocol after receiving the key is:

\[\Tr_{\H_{P_1}}\qty(E_{k^{P_2}}\qty(\tilde\rho^{n-1}(\tilde{P}_{2}, \hat{y})))\]

\noindent From Eq. %uation 
(\ref{eq:verification2}) (noting that %as taking the 
partial trace is a distance non-increasing operation): %with regard to distance):

\[\Delta\qty(\Tr_{\H_{P_1}}\qty(E_{k^{P_2}}\qty(\tilde\rho^{n-1}(\tilde{P}_{2}, \hat{y}))), \Tr_{\H_{P_1}}\qty(E_{k^{P_2}}(\rho^{n-1}_{ideal}(\hat{y})))) \leq \epsilon_{1}\]

\noindent  
We showed 
that the simulated view is $\epsilon_{2}$-close to the real view of $P_2$. 
\qed
\end{proof}

\bibliographystyle{abbrv}
\bibliography{biblio}

\end{document}